\documentclass[12pt,reqno]{amsart}
\usepackage[margin=2.5cm]{geometry}

\usepackage{amsmath,amssymb,amsfonts,amsthm,pifont,mathtools,,mathrsfs}
\usepackage[english]{babel}
\usepackage[utf8x]{inputenc}
\usepackage[T1]{fontenc}

\usepackage{graphicx}
\usepackage[colorinlistoftodos]{todonotes}
\usepackage[colorlinks=true, allcolors=blue]{hyperref}

\title{ Limit shapes for the dimer model}
\author{Nikolai Kuchumov}
\email{nikolai.kuchumov@gmail.com}

\newtheorem{theorem}{Theorem}

\newtheorem{proposition}{Proposition}
\newtheorem{claim}{Claim}

\newtheorem{remark}{Remark}

\newtheorem{corollary}{Corollary}

\newcommand{\RR}{\mathbb{R}}

\newcommand{\ZZ}{\mathbb{Z}}

\newcommand{\cN}{{\mathcal N}}

\newcommand{\pp}{\partial}

\newcommand{\cS}{\mathcal{S}}

\newcommand{\cA}{\mathcal A}

\newcommand{\fD}{\mathfrak D}
\newcommand{\GT}{\Gamma}
\newcommand{\OO}{\Omega}

\newcommand{\GG}{(\Gamma,\pp \Gamma)}
\newcommand{\GB}{(\Gamma,\chi)}
\newcommand{\hn}{\eta_n}
\newcommand{\avhn}{\bar{\eta_n}}

\newcommand{\Gn}{(\Gamma_n,\chi_n)}

\newcommand{\bb}{\begin{equation}}
\newcommand{\ee}{\end{equation}}

\newcommand{\norm}[1]{\left\lVert#1\right\rVert}

\begin{document}

\begin{abstract}
We prove the existence of a limit shape for the dimer model on planar \\ periodic bipartite graphs with an arbitrary fundamental domain and arbitrary periodic weights. This proof is based on a variational principle that uses the locality of the model and the compactness of the space of states.

\end{abstract}
\maketitle
\tableofcontents

\section{Introduction and Summary}
\label{0}
\subsection{Introduction}

This paper is devoted to the limit shape phenomena in lattice models of equilibrium statistical mechanics 

We use an example of a periodic planar dimer model with an arbitrary fundamental domain and arbitrary periodic weights.
Past results for arbitrarily planar domains were done for domino tilings (the fundamental domain correspondent to lattice $\ZZ^2$ in the dimer model) and weights equal to one \cite{CKP}.

The Dimer model is a stochastic model on finite graphs. We describe a set of configurations (dimer covers) in terms of discrete functions called height functions. When the size of the system is relatively small, the system behaves randomly. 
As the system grows, all height functions gather around exactly one continuous function, called \textit{the limit shape}
\footnote{ There are several lecture notes on modern studies of the limit shape phenomena, see \cite{AO} \cite{K} and references there.}.

The history of studies of limit shape phenomena goes back to the famous work by A.Vershik and S.Kerov on asymptotics for the Plancherel measure on Young diagrams \cite{Vershik-Kerov} using a variational principle and the following works,\cite{Vershik-Kerov2}, \cite{Vershik}.

\begin{figure}[h!] \label{YD}
\centering
\includegraphics[width=0.5\textwidth]{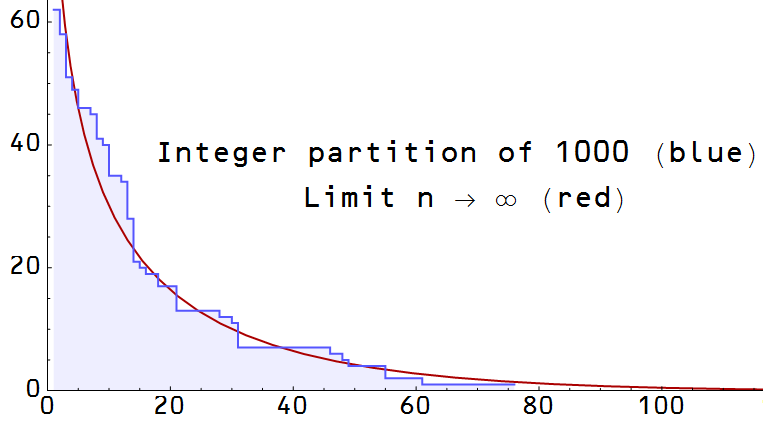}
\caption{The limit shape for the Young diagrams according to the Plancherel measure.}
\end{figure}

In more recent studies of limit shape phenomena the key example is random
domino tilings of Aztec diamond and the arctic circle theorem, \cite{JPS}, \cite{CEP} and \hyperref[Aztec]{Figure \ref{Aztec}.}. Later the technique was generalized to the arbitrary regions in \cite{CKP}.

\begin{figure}[h!]
\begin{minipage}[h]{0.39\linewidth}
\center{\includegraphics[width=0.4\linewidth]{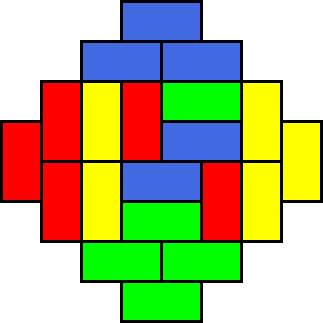} \\ Aztec diamond of size 4 }
\end{minipage}
\hfill
\begin{minipage}[h!]{0.59\linewidth}
\center{\includegraphics[width=0.7\linewidth]{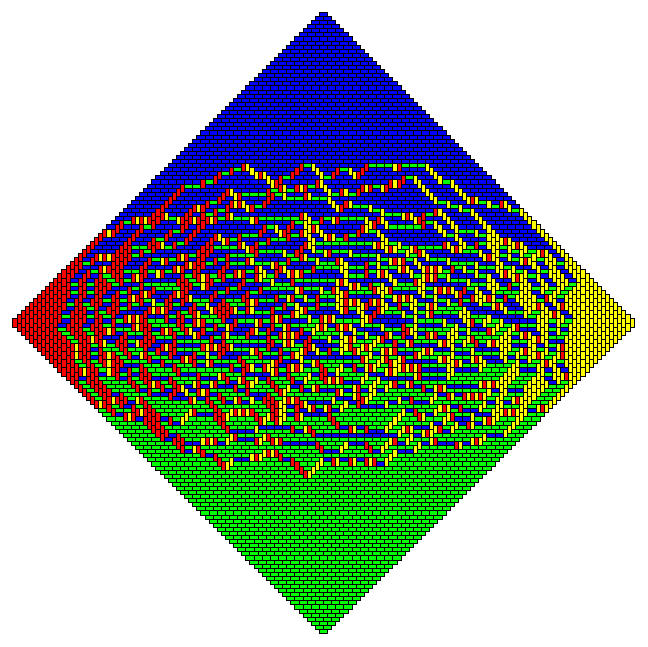} \\ Large scale behavior.}
\end{minipage}

\caption{ The arctic circle theorem on the left with frozen regions near the boundary. There are four types of dominos according to the standard bipartite structure on $\ZZ^2$ that are marked by four colors}
\label{ris:image1}
\label{Aztec}
\end{figure}
Then limit shape was found for 3d Young diagrams (also known as a plane partition) using Wolf crystal contraction, \cite{CK}, see \hyperref[Plane]{Figure \ref{Plane}}. 

\begin{figure}[h!]
\begin{minipage}[h]{0.39\linewidth}
\center{\includegraphics[width=0.6\linewidth]{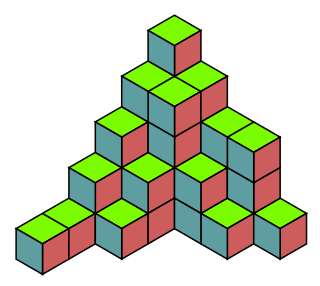} \\ An example of a plane partition.}
\end{minipage}
\hfill
\begin{minipage}[h!]{0.59\linewidth}
\center{\includegraphics[width=0.8\linewidth]{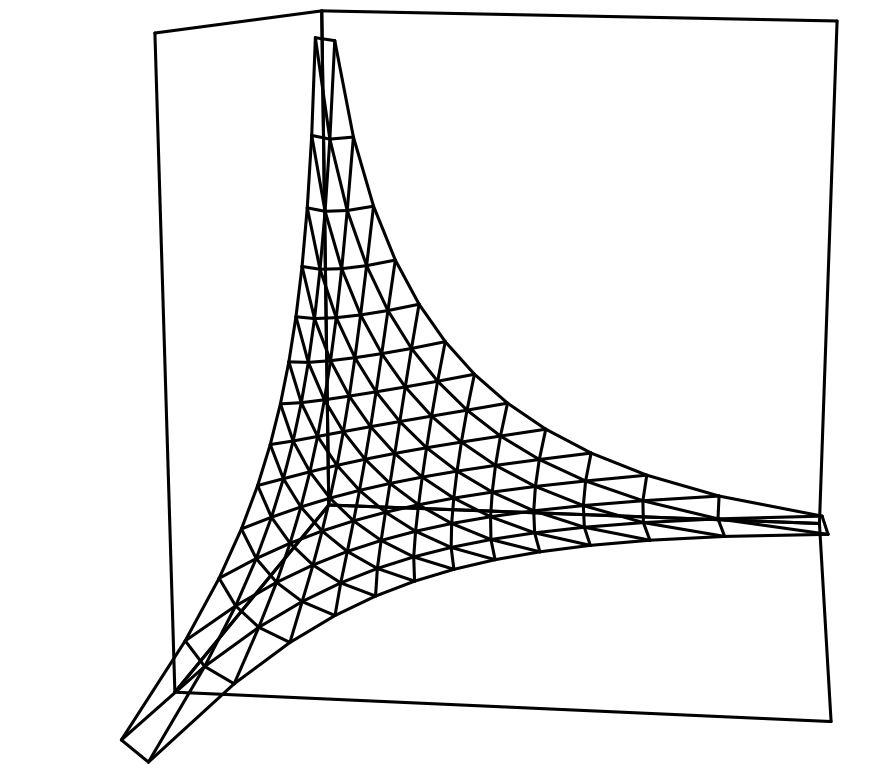} \\ The Wolf crystal around which plane partitions concentrate at large scales.}
\end{minipage}
\caption{The Wolf crystal for the plane partitions.}
\label{Plane}
\end{figure}
Later on there was a generalization of works on asymptotics of Young diagrams for deformed Plancherel measure,\cite{NO}.

Also, there was a work on a dimer model that generalized domino tilings (correspondent to the $\ZZ^2$ case of the dimer model) and built a connection to algebraic geometry,  \cite{KOS}, \cite{KO1}, \cite{KO2}.

Then there were works with purely algebraic methods that give alternative proofs of known theorems (i.e. the arctic circle theorem) and many new results, see \cite{BG}, \cite{BF} and references there.

In the last five years, several studies on limit shape phenomena of six vertexes were done, using algebraic methods for stochastic six-vertex model and a correspondent new type of integrable PDE was found \cite{RS}, \cite{BCG}.

\subsection{Summary}
The plan of the paper is as follows: first comes \hyperref[0]{introduction}, in the \hyperref[1]{second section} we discuss basic definitions of a dimer model, in the \hyperref[2]{third section} we define the probability distribution on a set of the dimer covers, then we define the \hyperref[3]{cutting rule} that plays the crucial role in the proofs. In the
next
section we discuss the $\ZZ^2$-periodic graphs with which we will work through the rest of the paper. In the \hyperref[4]{fifth section} we define the thermodynamic limit of the dimer model on
a plane and formulate the main result of the paper, \hyperref[The_limit]{The Limit Shape Theorem}.

Exactly the same proof works for the six vertex model as well. It will be presented in our next paper.

\subsection{Brief plan of the proofs}
The proof of the limit shape theorem consists of three parts, \textit{a variational principle}, which  we prove in the \hyperref[6]{seventh section} and in the \hyperref[7]{eighth section}. The second part is \hyperref[density]{The Density lemma} that we prove in the tenth section \hyperref[9]{Proofs of properties of height functions}. The third part of the proof is \hyperref[Functional]{the Surface Tension theorem} that is the key theorem in the proof of the main theorem.

A variational principle is a general analytic statement. We prove it under assumption of the \hyperref[Concentration_lemma]{Concentration lemma} and \hyperref[Functional]{the Surface Tension theorem}. The Concentration lemma is a probabilistic measure-theoretical statement which states that all height functions concentrate around an average height function.

In the second part we formulate and prove several properties of height functions, which lead us to \hyperref[density]{the Density lemma}. The idea of this part is that height functions are discrete analogs of Lipschitz functions, so the statements for the height functions and asymptotic height functions are almost the same including the proofs. Then we prove the
\hyperref[density]{Density lemma}
that states that for each asymptotic height function there exists a sequence of normalized height functions that
converges to it.

In the third part, we prove several auxiliary propositions that lead us to  \hyperref[Functional]{the Surface Tension theorem}. The idea is to approximate a square by a torus, then approximate a triangle by squares and approximate an arbitrary domain by triangles.

\section*{Acknowledges}
We are grateful to Institut Henri Poincaré for hospitality during the trimester "Combinatorics and interactions", where the part of the work was done. Also we would like to thank Anatoly Vershik for the excellent introduction to the limit shape phenomena and support, thank Nicolai Reshetikhin for the suggest to look at this problem and inspiring talk in PDMI RAS in 2017 and thank Vadim Gorin for the references. For useful comments and discussions we would like to thank Dmitry Chelkak, Fedor Petrov, Pavel Zatitskii, and Vladimir Fock. For the help in finding typos we would like to thank Masha Smirnova and Igor Skutsenya. Especially we would like to thank Pavlo Gavrilenko for reading the manuscript and Ilia Nekrasov without whose support and remarks the work would be impossible.

\section{Dimers on graphs with boundaries and height functions}
\label{1}
In this section we partly follow \cite{CR} and \cite{GK}.
\subsection{Dimers on graphs with boundaries}

A \textit{graph with a boundary} is a finite graph $\GT$ together with a set $\pp \GT$ of
valence one
vertices. We will refer to such verteces as \textit{boundary vertices} and the other vertices as \textit{internal}.

A \textit{dimer cover} $D$ on a graph with a boundary $\GG$ is a choice of edges of $\GT$, called dimers, so that each vertex, that is not a boundary vertex, is adjacent to exactly one dimer. Note that some of the boundary vertices may be adjacent to a dimer of $D$, and some may not. Dimer covers are also known as dimer configurations or perfect
matchings.
 
We will divide these partitions of boundary vertices into two groups - matched and non-matched and refer to them as \textit{boundary conditions} of dimer covers on $\GT$.


One can parametrize boundary condition by choosing a set of non-matched boundary vertices, let us use notation $\delta D$ for this set.

\begin{equation}
\delta D:=\{ v | v \in \pp \GT, v\text{ is not matched by }  D \}
\end{equation}

Let us call a set of the dimer covers on $\GG$ by $\mathfrak{D}(\GT)$ and denote the set of the dimer covers with fixed boundary conditions by
\bb
\fD(\GT;\delta D_0):=\{ D \in \fD(\GT) | \delta D= \delta D_0 \}
\ee

\subsection{Dimers on bipartite graphs}
\subsubsection{Bipartite structure}
 We will always assume that our graph is a \textit{surface graph}, that is a graph embedded
into
a compact oriented surface $S$ whose faces, i.e. the connected components of $S − \GT$, are contractible.
By embedding of graph with a boundary
into
surface $S$ we mean such embedding
 without self-intersections $\GT \xhookrightarrow[i]{} S$ that $i(\GT)\cap \pp S=i\left(\pp\GT\right) $
and the complement of $\GT$ $\backslash$ $\pp\GT$ in $S$ $\backslash$ $\pp S$ consists of open 2-cells. 

A \textit{bipartite structure} on a graph $\GT$ is a partition of its set of vertices into two groups, say blacks and whites, such that no edge of $\GT$ joins two vertices of the same group.

A bipartite structure induces an orientation on the edges of $\GT$, called the bipartite
orientation: simply orient all the edges from the white vertices to the black ones.

Now our graph is a cell complex and we will use usual boundary operator $\pp$ and standard notation for chain complex with the bipartite orientation. Also we will use following notation for expressions with the boundary operator.Let us denote by $[v]$ the 0-chain correspondent to a vertex $v$, similar for edges and faces. Suppose that $[e]_{12}$ is an edge between vertices $[v_1]$ and $[v_2]$, then we will use following notations:

\bb
\pp([e]_{12})=[v_1]-[v_2]=\sum_{i=1,2}sgn(v_i)[v_i] \in C_1(\GT,\ZZ),
\ee
where
$sgn(v) = 1$ is for the black vertex $v$, and
$sgn(v) = −1$ is for the white vertex.
\begin{figure}[h!]
\centering
\includegraphics[width=0.2\textwidth]{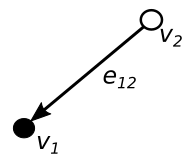}
\caption{An orientation of an edge}
\end{figure}

Equivalently, a bipartite structure can be regarded as a 0-chain

\begin{equation}
\beta = \sum_{v }sgn(v)[v] \in C_{0}( \GT , \ZZ)
\end{equation}
where the sum is over all vertices $v$ of $\GT$.

\subsubsection{Dimer covers}

\begin{figure}[h!]
\centering
\includegraphics[width=0.5
\textwidth]{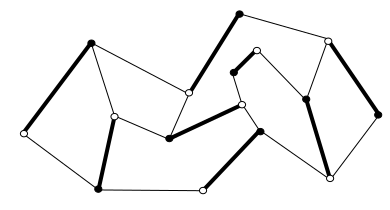}
\caption{An example of a dimer cover on a bipartite graph.}
\end{figure}

Using bipartite orientation, a dimer cover $D \in \fD\GG$ can now be regarded as a 1-chain with $\ZZ$-coefficients

\begin{equation}
D=\sum_{e \in D}sgn(e)[e]  \in C_1(\GT; \ZZ).
\end{equation}
The condition that $D$ is a dimer cover means that

\bb
\pp D=\sum_{v}sgn(v)[v]
\ee
where the sum is
taken
over all internal vertices and some boundary vertices of $\GT$\footnote{One can rewrite it as $\pp D=\sum_{v-black \ vertex}[v]-\sum_{v-white \ vertex}[v]$.}.

\subsubsection{Boundary conditions and decomposition cycles}
One can regard boundary conditions as an element in $C_0(\pp \GT;\ZZ)$, let $D$ be a dimer cover,
\begin{equation}
\delta D= \sum_{v \in \delta D}sgn(v)[v] \ .
\label{boundary_exp}
\end{equation}
Then one can notice that there is a relation between a dimer cover, its boundary conditions and bipartite structure on the graph. This relation is a set-theoretic statement with respect to bipartite orientation.
\begin{equation}
\beta=\sum_{v\text{ is not adjacent to } D}sgn(v)[v]+\sum_{v\text{ is adjacent to } D}sgn(v)[v]=\pp D + \delta D .
\end{equation}

Let $D, D^{\prime}$ be dimer covers with boundary conditions $\delta D$ and $\delta D^{\prime}$.
One can look at the difference $D-D^{\prime}$ in 1-chains for two dimer covers $D$ and $D^{\prime}$.

Then after substitution of (\ref{boundary_exp}) we get that $\partial ( D-D^{\prime} ) = \delta D - \delta D^{\prime}$. In case of equal boundary conditions $D-D^{\prime}$ is a 1-cycle. Otherwise it is true only in relative 1-chains, $C_0 (\Gamma,   \partial \Gamma; \mathbb{Z})$. 

Let us call $D-D^{\prime}$ \textit{decomposition cycles } of dimer covers $D$ and $D^{\prime}$.

\begin{figure}[h!]
\centering
\includegraphics[width=0.4\textwidth]{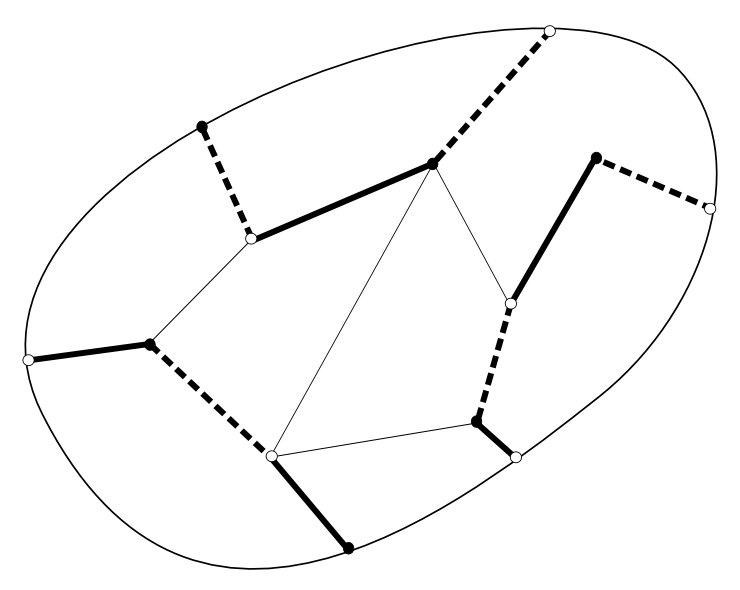}
\caption{An example of two dimer covers $D$ and $D^{\prime}$. Dimers from $D$ are in solid lines and dimers from $D^{\prime}$ in traced lines.}
\label{Decomposition}
\end{figure}

\subsection{Height function}
In this section we want to make a bijection between dimer covers of $\GT$ and some functions on faces of $\GT$. Somehow we want to map a dimer cover to its <<height>>. However, on non simply-connected surfaces there is no global height function and it is not actually a <<height>>, but rather <<The Penrose stairs>>, see \hyperref[Penrose]{Figure \ref{Penrose}.} 

\begin{figure}[h!]
\centering
\includegraphics[width=0.3\textwidth]{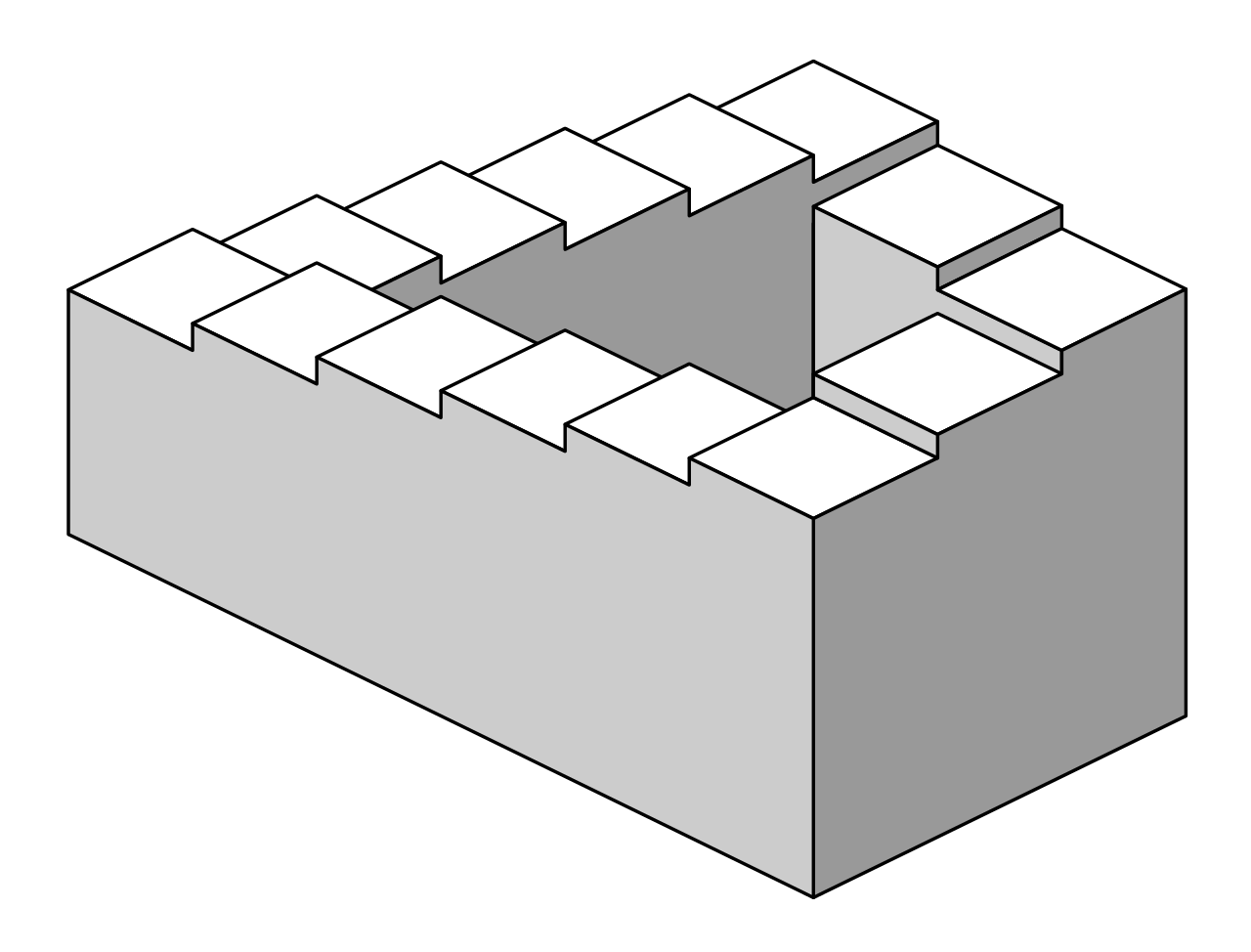}
\caption{<<The Penrose stairs>>.}
\label{Penrose}
\end{figure}

We will consider two cases --- graphs embedded into the torus and graphs embedded into the two-dimensional disk.

\subsubsection{Case of the disk}

Let $\GT$ be a graph with a boundary embedded into disk $\mathbb{D}^2$. We recall that $\GT$
induces
a cell decomposition of $\mathbb{D}^2$. We denote by \textit{boundary cells} such cells that contain boundary vertices.

Let us take two dimer covers, $D$ and $D^{\prime}$. Because $D-D^{\prime}$ is 1-cycle (rel $\pp \GT)$ there is an element $\sigma_{D,D^{\prime}} \in C_2(\GT,\pp\GT; \ZZ)$ such that $\pp\sigma_{D,D^{\prime}}=D-D^{\prime}$. Then let us define 2-cochain $h_{D,D^{\prime}}$ by the following formula:

\begin{equation}
\sigma_{D,D^{\prime}}=\sum_{f} h_{D,D^{\prime}}(f)[f] \in C_{2}(\GT,\pp\GT; \ZZ)
\end{equation}
where the sum
goes
over all faces of $\GT$.
We can
regard
cocycle $h_{D, D^{\prime}}$ as a function on faces of $\GT$. We will call such functions
as \textit{height functions}. 

One can notice that a height function is simply a function of level for decomposition cycles, so it is uniquely defined by $D, D^{\prime}$ up to an additive constant. Hence, one can normalize all height functions by setting $h_{D,D^{\prime}} (f_0) = 0$ for some fixed face $f_0$. Note that a height function satisfy Lipschitz condition in some sense. For any two faces $y$ and $x$ values of a height function at this faces differ at most by the length on the dual graph between the faces.

\begin{figure}[h!]
\centering
\includegraphics[width=0.5\textwidth]{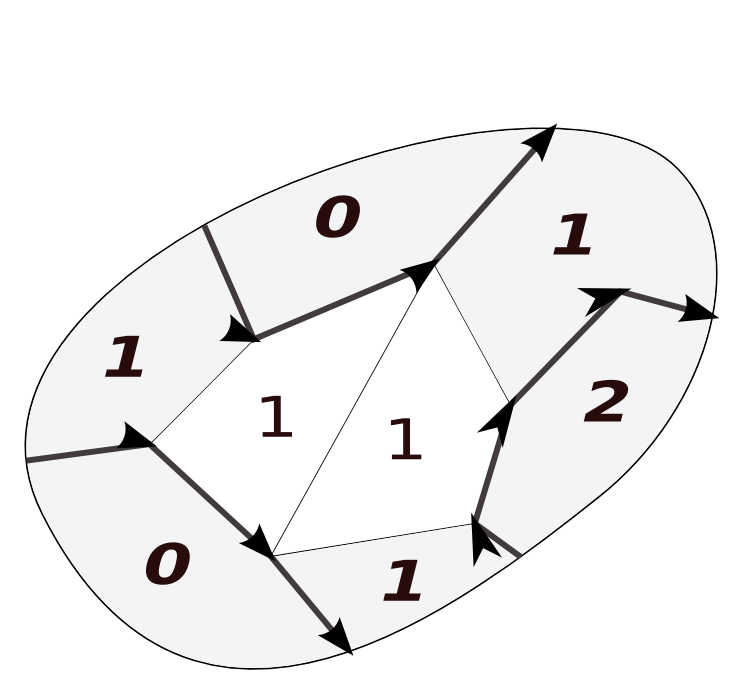}
\caption{Decomposition cycles for Dimer covers from  \hyperref[Decomposition]{Figure \ref{Decomposition}} and correspondent height function.}
\end{figure}

One can think about boundary conditions for dimer covers in terms of height functions. Let us look at two dimer covers $D$, $D^{\prime}$ and
their
height function $h_{D,D^{\prime}}$. A \textit{boundary height function} of $D$ and $D^{\prime}$ is a restriction of the height function to boundary faces. Fix dimer cover $D^{\prime}$ and look at different $D$. It is clear that one can reconstruct boundary conditions of $D$ using boundary height function and $\delta D^{\prime}$. Simply because the height function changes across edges where $D$ and $D^{\prime}$ are different from each other. So it is sufficient to change a boundary condition for $D^{\prime}$ along every edge where the height function changes.

Also note that for any three dimer configurations $D$, $D^{\prime}$ and $D^{\prime \prime}$ on $\GT$, the following cocycle equality holds:

\bb
h_{D,D^{\prime}} + h_{D^{\prime}
,D^{\prime\prime}} = h_{D,D^{\prime \prime}} .
\ee

Later on we will work with height functions $h_{D,D^{\prime}}$ for a fixed reference dimer configuration $D^{\prime}$. For regular graphs, there is an alternative definition of height functions called \textit{absolute height functions} which we review in \hyperref[abs]{the appendix}.

\subsubsection{Case of a torus}

In case of a torus height function it can be defined as a function only locally.because of monodromy along not simply-connected cycles\footnote{For the case of not simply-connected region the height function is not a function, but rather a <<section of $\ZZ$ bundle>>. It defines a function locally on every simply-connected component. So it is uniquely defined up to an action of $\pi_1(S)$, which adds for each loop monodromy along it.}.
Note that in this case there
are
no boundary vertices. Look at \hyperref[height_torus]{Figure 6} for an example.
\begin{figure}[h]
\begin{minipage}[h]{0.49\linewidth}
\center{\includegraphics[width=1\linewidth]{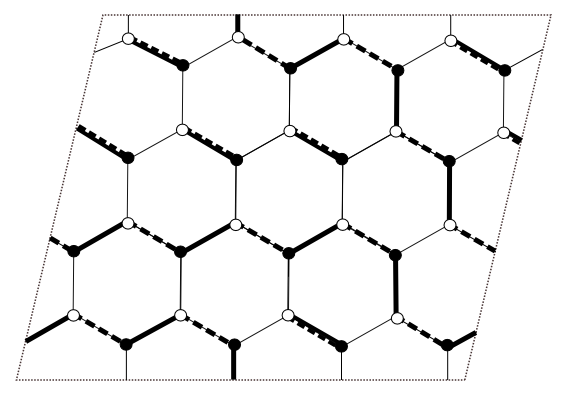} \\ }
\end{minipage}
\hfill
\begin{minipage}[h]{0.49\linewidth}
\center{\includegraphics[width=1\linewidth]{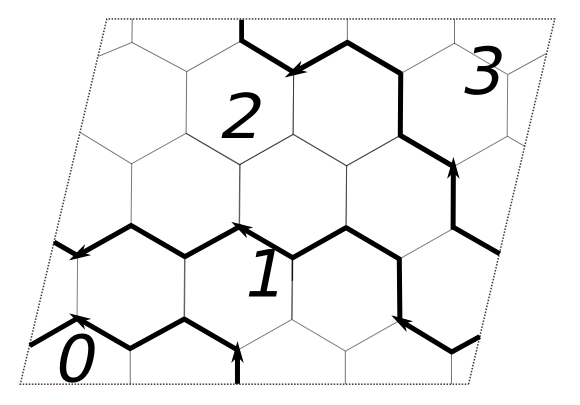} \\ }
\end{minipage}
\label{height_torus}
\caption{An example of a height function for two dimer covers on a torus. Note that is has monodromy $2$ along vertical cycles and $1$ along horizontal ones.}
\end{figure}

\subsection{Newton polygon}

Because $[D]-[D^{\prime}]$ is 1-cycle, it defines a homology class. Due to $H_1(\mathbb{T}^2,\ZZ) \simeq \ZZ^2 $ the homology class can
be
identified with a pair of integers $(s_D,t_D)$ after the choice of a basis $(\gamma_1,\gamma_2)$ in $H_1(\mathbb{T}^2,\ZZ)$.

Each $s_D$ and $t_D$ are intersection numbers of $[D]-[D]^{\prime}$ with $\gamma_1$ and $\gamma_2$, respectively.
In other words it is a monodromy of a
corresponding
height function along the cycles.

We will call this pair of integers $(s_D,t_D)$ a \textit{slope} of dimer cover $D$. Note that it depends on $D^{\prime}$ which is fixed.

In case of a graph $\GT$ embedded into a torus, a set of slopes of the all dimer covers is denoted by $\mathcal{S}(\GT)$,

\bb
\mathcal{S}(\Gamma):=\{ (s_D,t_D) \in \ZZ^2, D\in \mathfrak{D}(\GT) \}
\ee
It is a finite set of points in $\ZZ^2$ and it is uniquely defined up to a change of $D^{\prime}$ that shifts the set of slopes.

\noindent The Newton polygon for the $\GT$ is a convex
hull of the set of slopes,

\bb
N_{\Gamma}:=Conv(\mathcal{S}(\GT))
\ee

For example, see \hyperref[fig:Newton]{Figure \ref{fig:Newton} } for the case of square grid.

Note that $N_G$ is defined up to a linear shift, so if $0 \notin N_G$ we can make a proper linear shift to fix it.

\begin{figure}[h]
\begin{minipage}[h]{0.59\linewidth}
\center{\includegraphics[width=0.4\linewidth]{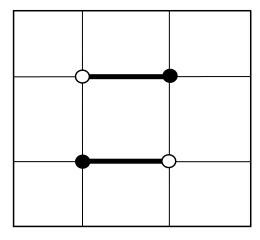} \\ }
\end{minipage}
\hfill
\begin{minipage}[h]{0.40\linewidth}
\center{\includegraphics[width=0.8\linewidth]{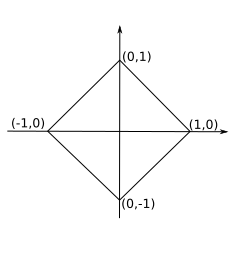} \\ }
\end{minipage}
\caption{An example of the Newton polygon for square grid on the right and
corresponding
$D^{\prime}$ on the left}
\label{fig:Newton}
\end{figure}

\section{Boltzmann distribution on dimer covers}

\label{2}
Let $\GG$ be a graph with a boundary.
A \textit{weight system} on $\GG$ is a map from the set of dimer covers $\mathfrak{D}(\GT)$ to positive numbers $\RR_{>0}$.
A weight system $w$ defines a Boltzmann distribution (also known as a Gibbs measure) on the set of dimer covers.
For $D \in \mathfrak{D}(\GT)$ let us define its probability by \hyperref[Gibbs]{ (\ref{Gibbs})},
\bb
\mathbb{P}(D):=\frac{w(D)}{Z(w;\GT)}
\label{Gibbs}
\ee
where $D \in \fD(\Gamma)$ and $Z(w;\GT)$ is a normalization constant called the partition function:
\bb
Z(w;\GT)=\sum_{D \in \mathfrak{D}(\GT)}w(D).
\ee

We shall focus on a particular type of weight systems called \textit{edge weight systems}.
Let us assign to each edge $e$ of $\GT$ a positive real number $w(e)$ called the weight of the
edge $e$. The associated edge weight system on $\mathfrak{D}(\GT)$ is given by
\bb
w(D)=\prod_{e \in D}w(e)
\ee
where the product
goes
over all edges contained in $D$.

In statistical mechanics, these weights are called Boltzmann weights. Their physical meaning can be expressed by the following formula:
\bb
w(e) = \exp \left( \frac{-E(e)}{kT}\right),
\ee
where $E(e)$ is the energy of the dimer occupying the edge $e$, $T$ is the absolute temperature and $k$ is the Boltzmann constant.

\begin{remark}
Boltzmann distribution on dimer covers
induces
a probability distribution on the set of
corresponding
height functions. Probability of a height function $h_D$ is just probability of the dimer cover $D$.

We will denote by $\bar h$
corresponding
expectation value of height function and by $\bar h(v)$ its value at face $v$.

\end{remark}

\section{The cutting rule for the dimer model}
\label{3}
\subsection{Graphs with boundary conditions}
Let $\GG$ be a graph with a boundary and let us fix a boundary height function $\chi$. Then we will call such a pair a \textit{graph with a boundary condition} and denote by $\GB$.

 Let $\fD \GB$ be a set of dimer covers on $\GT$ with a boundary height function $\chi$. Let us denote a set of height functions with the given boundary condition $\chi$ by $\mathcal{H}(\Gamma,\chi)$. Usually we will be interested in partition functions associated with $\GB$.

\bb
Z \GB:=\sum_{D \in \fD \GB}{w(D)}
\ee

It is straightforward that the partition function of $\GG$ is the sum of $Z(\GB)$ over all boundary height functions:

\bb
Z\GG =\sum_{\chi}{Z{\GB}}
\ee

\subsection{The Cutting Rule}

Let $\GG$ be a graph with a boundary, and let us fix an edge $a\in E(\GT)$. Let us denote by $(\GT_{a}, \pp\GT_{a})$ a graph with a
boundary obtained from $\GG$ as follows: cut the edge $a$ into two edges $a_1$ and $a_2$, and set $\pp \GT_{a}=\pp\GT \cup \{ v_1, v_2 \}$,
where $v_1$ and $v_2$ are the new
valence one
vertices adjacent to $a_1$ and
to $a_2$, respectively.
See on \hyperref[fig:cutting1]{Figure \ref{fig:cutting1}}. 

Note that there is a natural map $\phi_a : \fD\GG \mapsto \fD (\GT_{a},\pp \GT_{a}) $ that simply cuts in two the dimer of $D$ that belongs to the edge $a$.

\begin{figure}[h!]
\centering
\includegraphics[width=0.45\textwidth]{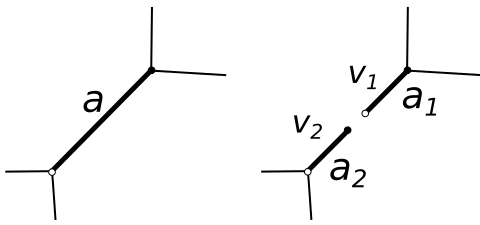}
\caption{An example of the cutting procedure}
\label{fig:cutting1}
\end{figure}

Then we want to find a <<pushforward>> $\phi_{*}$ of weight systems with respect to the cutting procedure described above. Thus we want to construct a weight system on $\fD(\GT_{a},\pp \GT_{a})$  so that it has
the
following property:
\bb
w(D)=\phi_{*}w(\phi(D))
\label{pushforward-weight}
\ee 
It means that all dimer covers that come from $\GT$ will have the same weight.

Then the partition function of $\GG$ can be identified with a part of the partition function of $(\GT_{a},\pp\GT_{a})$ that comes from $\GG$.

\bb
Z\GG=\sum_{D \in \fD \GG}{w(D)}=\sum_{D \in \fD \GG}{\phi_{*} w (\phi (D) )}=
\sum_{D^{\prime} \in \phi (\fD \GG)}{\phi_{*} w (D^{\prime})}
\label{Cutting_property}
\ee

There is a general way to construct $\phi_{*}w$:

For all edges that are not coming from cutting the edge $a$ weights remain as they are on $\GT$.

For the case of edges $a^{\prime}$ and $a^{\prime\prime}$, that come from cutting edge $a$, it is sufficient to take such weights that $w(a)=w(a_1)w(a_2)$. That way the condition $\ref{pushforward-weight}$ will be true.

General way to do it is to take $w(a^\prime)=\sqrt[]{w(a)}t$ and $w(a^{\prime \prime})=\sqrt[]{w(a)}t^{-1}$ for $t>0$. Later on we will use $t:=\sqrt[]{w(a)}$ to make $w(a^\prime)=w(a)$ and $w(a^{\prime\prime})=1$.

Note that due to locality of the cutting procedure cuts of different edges
commute
in a sense that we obtain the same graph with a weight system. This means that the cutting procedure can be extended for a family of edges $\mathcal{A}$. Just iterate the procedure for all edges $a\in \mathcal{A}$. The corresponding map $\phi$ will be just a composition of maps $\phi_{a}$ for edges $a \in \mathcal{A}$, the same for $\phi_{*}$ and the \hyperref[Cutting_property]{Cutting rule} 
 will still holds. 

\subsection{The cutting rule for surface graphs}
\subsubsection{The cutting of graphs with boundaries}

Let $\GB$ be a graph with a boundary condition embedded into surface $\mathcal{S}$ and $\rho$ be a simple curve in $\cS$ that is “in general
position” with respect to $\GT$, in the following sense:
\begin{itemize}
\item it is disjoint from the set of vertices of $\GT$;
\item it intersects the edges of $\GT$ transversally;
\item its intersection with any given face of $\GT$ is connected.
\end{itemize}

Let $\cS_\rho$ be the surface with the boundary obtained by an open cutting $\cS$ along $\rho$.
Also let $\GT_\rho := \GT_{A(\rho)}$ be a graph with a boundary obtained by cutting $\GG$ along the set $A(\rho)$ of edges of $\GT$ that intersect $\rho$.

Obviously, $\GT_\rho$ is a surface graph with a boundary. We will say that it is obtained by the cutting $\GT$ along $\rho$.Let us denote by $\phi_\rho$ the map of the dimer covers and by $w_\rho$ the weight system on $\GT_\rho$ obtained from the weight system on $\GT$.

\begin{figure}[h!]
\begin{minipage}[h!]{0.40\linewidth}
\center{\includegraphics[width=1.1\linewidth]{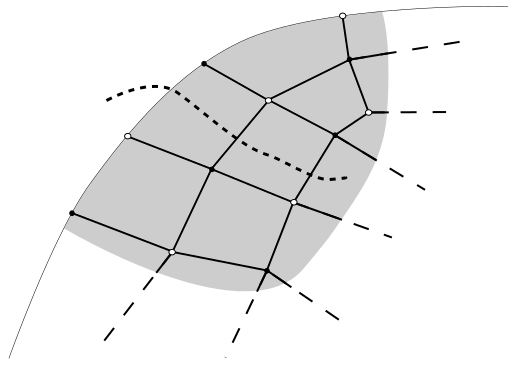} \\ }
\end{minipage}
\hfill
\begin{minipage}[h!]{0.40\linewidth}
\center{\includegraphics[width=1.1\linewidth]{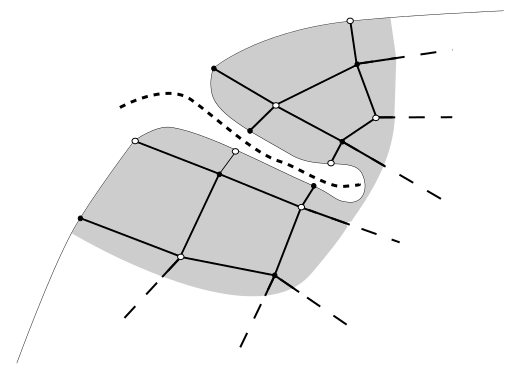} \\ }
\end{minipage}
\caption{An example of a cutting along a curve }
\end{figure}

\bb
Z\GB=\sum_{D \in \fD \GB}{w(D)}=\sum_{D^{\prime} \in \phi_\gamma (\fD \GB)}{w_\rho (D^{\prime})}
\ee

\subsubsection{Cutting graphs with boundary conditions}

The next step is to cut graphs with
boundaries, to do it we need to set a boundary height function along $\GT_\rho$. Basically we want to describe the set $\phi ( \fD \GB  )$ to rewrite the cutting rule in terms of height functions.

Suppose that we have a graph with a boundary condition $\GB$.
There are two types of faces on $\GT_\rho$: each face is either an \textit{old face} that comes from a face in $\GT$ or a \textit{new face} obtained from cutting a face of $\GT$ intersected by the curve $\rho$.

For old faces we can leave old values of $\chi$.
For new faces one may take boundary conditions that are obtained by cutting a dimer cover on $\GT$ which we will parametrize by a boundary height function $\chi_\rho$.

In terms of boundary height function it means the following:
There are faces that intersect with $\rho$, we denote the set of such faces by $\mathfrak{F}(\rho)$. There is a natural map $\psi$ from $\mathfrak{F}(\rho)$ to pairs of boundary faces on $\GT_{\rho}$ that are obtained by cutting faces from $\mathfrak{F}(\rho)$. Then we have one boundary value for a face $f \in \mathfrak{F}(\rho)$ and a pair of values for $\psi(f)=(f_1,f_2)$.
The condition that $\chi_\rho$ is obtained by cutting along $\rho$ means that $\chi(f)=\chi_\rho(f_1)= \chi_\rho(f_2).$ \footnote{In case $\rho$ cuts the graph into two graphs the <<equality>> means that the boundary height functions agree up to an additive  constant.} See \hyperref[Cutting_boundary]{Figure \ref{Cutting_boundary1}}

\begin{figure}[h!]
\begin{minipage}[h!]{0.45\linewidth}
\center{\includegraphics[width=1\linewidth]{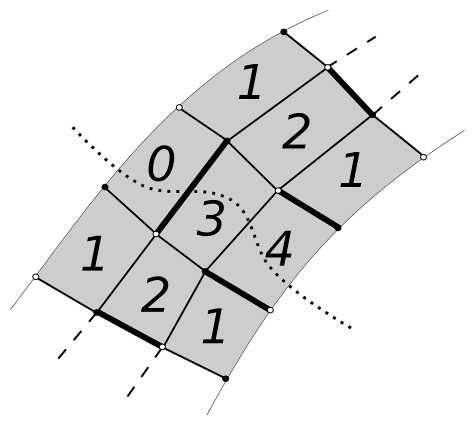} \\ }
\end{minipage}
\hfill
\begin{minipage}[h!]{0.45\linewidth}
\center{\includegraphics[width=1\linewidth]{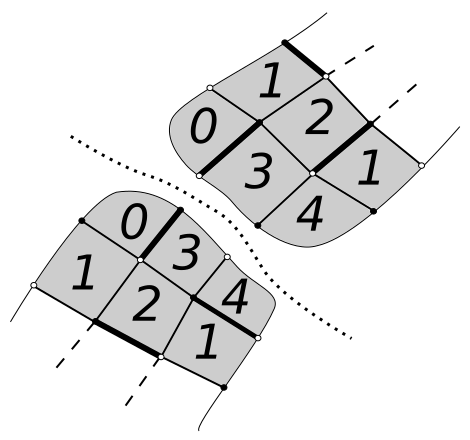} \\ }
\end{minipage}
\caption{An example of a cutting of boundary condition.}
\label{Cutting_boundary1}
\end{figure}

\begin{proposition}[The Cutting rule]
Suppose we cut a graph with a boundary condition $\GB$ into two graphs $\GT^{\prime}$ and $\GT^{\prime \prime}$.Then we can express the partition function of $\GB$ in terms of $\GT^{\prime}$ and $\GT^{\prime \prime}$.
\bb
Z\GB=\sum_{\chi_\rho}{Z(\GT^{\prime},\chi^{\prime}_\rho)Z(\GT^{\prime \prime},\chi^{\prime \prime}_\rho)}
\label{Cutting_boundary}
\ee
Here $\chi_\rho$ means a pair of the boundary height functions $(\chi^{\prime}_\rho,\chi^{\prime \prime}_\rho)$ that are obtained from $\chi$ by cutting along $\rho$. 
\end{proposition}

\begin{remark}
The cutting rule has the following combinatorial explanation: 
We need to sum up over all dimer covers. One way to do it is to calculate partition functions $Z(h(\rho))$ with given height function $h(\rho)$ along $\rho$. And then to sum up $Z(h(\rho))$ over all $h(\rho)$. The result is the same because we just permute terms in a finite sum.

Then, we can interpret each $Z(h(\rho))$ as a product of two partition functions because $\rho$ cuts $\GT$ into two graphs where dimer covers are independent. Thus until we fix the boundary condition along $\rho$ to be $h(\rho)$, we are calculating the original sum $Z(h(\rho))$.
\end{remark}

\section{Periodic graphs}
\label{4}
\subsection{Periodic graphs}

In this section we follow 2nd and 3rd paragraphs of \cite{KOS} .

Let $G$ be a $\ZZ^2$-periodic planar graph. By this we mean that $G$ is a bipartite graph embedded
into
the plane $\mathbb{R}^2$ (we denote the embedding by $\phi$)
so that translations in $\ZZ^2$ act by color-preserving isomorphisms of $\phi(G)$ – isomorphisms which preserve bipartite structure: map black vertices to black vertices and white to white.

Let $G(n)$ be the quotient of $G$ by the action of $n\ZZ^2$. It is a finite bipartite graph on a torus. $G(1)$ is called a fundamental domain.
An example is $\ZZ^2$ itself with fundamental domain on \hyperref[fig:Newton]{Figure \ref{fig:Newton} }.

Later on we will need
finite subgraphs in $G$ and call them
planar periodic graphs. $G$ will be fixed in all following statements and theorems. Note that periodic planar graphs are also graphs with boundaries. We will denote Newton polygon corresponding to $G(1)$ by $N_G$.

We will fix a weight system on $G(1)$ and then continue it to a weight system on $G$ by periodicity.\footnote{Weight systems on $G(n)$ can be obtained by the same construction.}

\subsection{Thermodynamic limit of the dimer model on the torus.}

Let us look at $G(n)$, one can calculate its partition function $Z(G(n))$ in limit as $n \to \infty$.

Because actual partition function grows rapidly it is convenient to normalize it, 
\begin{equation} F:=\lim_{n \to \infty}  n^{-2}\log(Z(G(n))). \end{equation}
This limit exists and the answer is the following \cite{KOS}:

\begin{equation} F=\iint\limits_{ \vert z \vert = \vert w \vert = 1} \frac{1}{(2 \pi i^2)} \log P(z,w) \frac{\,dz\,dw}{zw},\end{equation}

Here $P(z,w)$ is a Laurent polynomial in $z$ and $w$, with Newton polygon  $N_{G}$. It depends only on the fundamental domain and the weight system on it.

For the case of square lattice with all weights equal 1 the polynomial is the following:

\begin{equation} P(z,w)=\frac{(1+z)^2}{z}+\frac{(1+w)^2}{w} \end{equation}

\subsection{Surface tension}
For fixed $(s, t) \in \mathbb{R}^2$ we denote by $\mathfrak{D}_{s,t} (G(n))$ the set of dimer covers on $G(n)$, that have a slope $(\lfloor{ns}\rfloor,\lfloor {nt}\rfloor)$. \footnote{$\mathfrak{D}_{s,t}(G(n))$ is not empty for $(s,t) \in N_G$}.

Consider the normalized partition function of dimer covers with the fixed slope $s,t$

\begin{equation}
Z_{s,t}(T_n) = n^{-2} \log \sum_{D \in \mathfrak{D}_{s,t}(G(n))} w(D).
\label{surface3}
\end{equation}
This partition function is called the surface tension. It has a limit as $n \to \infty$ and the answer is
\begin{equation}  \lim_{ n \to \infty }n^{-2} \log(  Z_{s,t}(G(n))) = \sigma(s,t) \end{equation}

Here $\sigma$ is minus Legendre transform of Ronkin function of $P(z,w)$,
\begin{equation}R(B_x,B_y):=\iint_{|z|=e^{B_x}, |w|=e^{B_y}} {\log P(z,w)}\frac{dzdw}{zw}\end{equation}

From the corollary 3.7 from \cite{KOS}, $\sigma$ is a concave function \footnote{In our notations of $\sigma$ we follow \cite{CKP}, so they differ from the notations in \cite{KOS} by the sign of $\sigma$.}.

\section{Thermodynamical limit of the dimer model on the plane}
\label{5}
Let $\OO$ be a compact connected simply-connected domain in $\RR^2$. In this section we will look at sequences of increasing graphs embedded into $\OO$ and formulate the limit shape theorem for it. 

In following paragraphs we assume that all functions are bounded real-valued functions on $\Omega$, we will denote a set of this functions on $\OO$ by $\mathcal{B}(\Omega)$. We recall that there is a standard norm on $\mathcal{B}(\Omega)$, $\norm{f}:=\max_{x \in \Omega}{|f(x)|}.$

Suppose that we have a graph with a boundary $\GG$ embedded into $\Omega$. Then  height functions  on $\GG$ can be treated as bounded functions on $\Omega$ assigning its value at point $x$ to be $h(f)$ if $x$
is
contained in one face $f$ and average of values $h(f)$ over all faces $f$ containing $x$.
\subsection{Asymptotic height functions}

Due to the Lipschitz condition
on
height functions one might expect that in a limit they will converge to continuous Lipschitz functions.

Let $Lip(\Omega)$ be the space of all Lipschitz functions on $\Omega$.
\bb
Lip(\Omega):=\{ f \in \mathcal{B}(\Omega) |  \exists C>0 : \  \forall x,y \in \Omega \\ \ |f(x)-f(y)|\leq C |x-y| \}.
\ee
For each function $f$ such a minimal constant $C$ from the definition is called the Lipschitz constant for the function $f$ and $f$ is called $C$-Lipschitz in this case. Lipschitz functions are continuous and
differentiable almost everywhere due to Rademacher's theorem.

Let us define a set of asymptotic height functions by
$$\mathscr{H} (\Omega):=\{ f \in Lip(\Omega) \ | \ \nabla f \in N_G \text{ almost everywhere} \}.$$
Later on we will prove \hyperref[density]{The Density Lemma}, that justifies this definition. 

Let us take an asymptotic height function $\chi$ that is fixed on $\pp \Omega$. We will call such a pair $(\Omega,\chi)$ \textit{a domain with a boundary condition} and we will look at asymptotic height functions that coincide with $\chi$ on $\pp \Omega$, let us denote the set of such asymptotic height functions by $\mathscr{H}(\Omega,\chi)$.

Note that $N_G$ is defined up to a linear shift, so if $0 \notin N_G$ we can make a proper linear shift to fix it.

\subsection{Approximations of domains}
Let $d_n: \RR^2 \to \RR^2 $ be a dilatation $d_n: (x,y) \mapsto (n^{-1} x,n^{-1} y).$ 
Let us look at rescaled embedding $\phi_{n}:=d_n \circ \phi$, and denote rescaled $G$ by $G_{n}:=\phi_{n}(G)$.
Also we need to rescale height functions on $\Omega_n$ by a factor $n^{-1}$, let us denote them by $\hn$ and its expectation values by
$\avhn$. We will call them 
\textit{normalized height functions}.

Let $(\Omega,\chi)$ be a domain with a boundary condition. Then we will call a sequence of graphs with boundary conditions $\{ \GT_{n},\chi_{n} \}$ \textit{an approximation} of $(\Omega,\chi)$ (where $\{\chi_n\}$
are
normalized boundary height functions) if

\begin{enumerate}
\item $\GT_{n} \subset G_n \cap \Omega$
\item each $\GT_{n}$ admits at least 1 dimer cover with normalized boundary height function $\chi_{n}.$
\item $\norm{\chi_n- \chi} \to 0$ as $n \to \infty$.
\item $\Gamma_{n}$ tends to $\OO$ with respect to Hausdorff distance\footnote{Here $d_H$ means Hausdorff distance, $d_{H}(X,Y)=\inf\{\epsilon \geq 0\,: \ X\subseteq Y_{\epsilon }\ \text{ and } \ Y\subseteq X_{\epsilon }\}$, where  $X_{\epsilon }$ is $\epsilon$-neighborhood of $X$.}, $d_H(\Omega,\Gamma_{n}) \to 0,$ as $n \to \infty$ \footnote{Actually we will assume that the sum of areas of faces of $\Gamma_{n}$ converges to the euclidean area of $\Omega$.}.
\end{enumerate}

\subsection{The Limit Shape Theorem}

\begin{theorem}[The limit shape theorem]
Let $(\Omega,\chi)$ be a domain with a boundary condition and $(\Omega_{n},\chi_{n})$ be an approximation of $(\Omega, \chi)$, then

$$\lim_{ n \to \infty } {n^{-2} \log Z \left( \Omega_{n},\chi_{n} \right) } = \iint_{\Omega}{\sigma (\nabla g) dx dy}$$
where $g$ is the maximizer of the functional  $\mathcal{F} (h) := \iint_{\Omega}{\sigma (\nabla h) dx dy}$ on the set $\mathcal{H}(\Omega,\chi)$.

Moreover, let $\hn$ be a random height function on $( \Omega_{n}, \chi_{n} )$. Then we have the convergence in probability for $\hn$, that is for each $c>0$

$$\mathbb{P} \left( \norm{\hn - g}>c\right) \to 0 
\text{ as } n \to \infty.$$
So all the height functions converge to the <<limit shape>> $g$ pointwise in probability.
\label{The_limit}
\end{theorem}

\section{The proof of the variational principle}
\label{6}
In this section we prove the \hyperref[The_limit]{The Limit Shape Theorem} under assumptions of \hyperref[Concentration_lemma]{the Concentration Lemma} and \hyperref[Functional]{the Surface Tension Theorem}. First we prove that all height functions converge to the limit shape, then we show that the partition function  is localized around the limit shape and finally notice that the limit shape is the maximizer of the surface tension functional.

\subsection{Convergence of height functions to the limit shape}
Let $(\Omega_n,\chi_n)$ be an approximation of a domain with a boundary condition $(\Omega, \chi)$.

Consider the sequence of average height functions, $\{ \avhn \}$. By the $\hyperref[density]{density \  lemma}$ there is a sequence of asymptotic height functions $\{g_n\}$, such that $\norm{g_n-\avhn}\leq \frac{C}{n}$. Due to compactness of $\mathcal{H}(\Omega,\chi)$ $\{g_n\}$ has a convergent subsequence, let us denote its limit by $g$. Without loss of generality we suppose that convergent subsequence is $\{g_n\}$ itself. Now we denote by $B_C(\avhn)$ balls of radius $C$ around $\avhn$.

From the $\hyperref[Concentration_lemma]{Concentration \ lemma}$ and its corollary we get that 

\bb
\mathbb{P}\left( \hn \notin B_{C}(\avhn)\right)\leq n^2\exp(-K n C^2/2)
\ee
for some real positive constant $K$.

Consider balls of the radius $\frac{C}{n^{1/3}}$ around $\{\avhn\}$, let us call them $B_{1/3}(\avhn)$. Note that $\{g_n\}$ lie in these balls due to $\frac{C}{n^{1/3}} \geq \frac{C}{n}$ and similar bound takes place. The difference between $B_C(\avhn)$ and $B_{1/3}(\avhn)$ is that the new sequence of balls becomes smaller as $n \to \infty$ and all sequences of height functions lying in the balls converge to the same limit as $g_n$ which is $g$.
\bb
\mathbb{P}\left( \hn \notin B_{1/3}(\avhn)\right)\leq n^2\exp(-K {n^{1/3}}C^2/2) \to 0 \text{ as } n \to \infty
\label{conc_lim}
\ee
for $K>0$ by the \hyperref[density]{Density \ lemma}.

So all height functions converge pointwise in probability to $g$.

\subsection{The localization of the partition function}
Let us recall an expression for probability that a random normalized height function lies in a set $\cS$, $\mathbb{P}(\hn \in \cS)=\frac{Z(\Omega_n |\cS)}{Z(\Omega_n,\chi_n)}$ where $Z(\Omega_n| \cS)$ is a partition function where we sum only dimer covers with normalized height functions lying in the set $\cS$. Let us take $\cS=B_\delta(g)$ for sufficiently small fixed $\delta>0$ that we sent to $0$ in the end of the proof.

\bb
1=\mathbb{P}(\hn \in B_\delta (g))+\mathbb{P}(\hn \notin B_\delta (g))
\ee

\bb
\lim_{n \to \infty}{\frac{Z(\Omega_n|\delta, g)}{Z(\Omega_n)}}=1-\lim_{n 
\to \infty}\mathbb{P}(\hn \notin B_\delta (g))
\label{division}
\ee

 Then ($\hyperref[conc_lim]{ \ref{conc_lim}}$) says that probability in the right hand side vanishes as $n \to \infty$. Taking logarithm of (\ref{division}) and multiplying by $n^{-2}$ we get that 
\bb
\lim_{n \to \infty} n^{-2}\log(Z(\Omega_n))=\lim_{n \to \infty} n^{-2}\log(Z(\Omega_n|\delta, g)).
\ee
One can express the right hand side by \hyperref[Functional]{the Surface Tension Theorem}, 

\bb
\lim_{n \to \infty} n^{-2}\log Z(\Omega_n)=\lim_{n \to \infty} n^{-2}\log Z(\Omega_n|\delta,g)=\mathcal{F}(g)
+o_\delta(1)
\label{converg}
\ee
where $o_\delta(1)\to 0$ as $\delta\to 0$, so after taking limit $\delta \to 0$ we get the second part of the proof.

\subsection{The Surface tension functional and the limit shape}

By proposition 2.4 from \cite{CKP} there exists the unique maximizer of the surface tension functional on $\mathcal{H}(\Omega,\chi)$. Then $g$ is the maximizer of the surface tension functional, because $\mathcal{F}(g)$ is proportional to $\lim_{n \to \infty} n^{-2}\log Z(g,\delta)$, but all the height functions concentrate around $g$ by the $\hyperref[conc_lim]{ (\ref{conc_lim})}$.

\section{The concentration lemma}
\label{7}

In this section we follow section <<6.2. Robustness.>> from \cite{CEP}.

\begin{claim}
\label{Concentration_lemma}
Let $(\Gamma, \partial \Gamma)$ be a graph with a boundary. Let us fix a boundary condition and take a face $v$ on $\Gamma$.

Then $h (v)$ is a random variable and the following estimate for an expectation value is true

Let us take $a>0$, then
 \bb
 \mathbb{P} \left( | h(v)-   \bar h(v)|> a \cdot \sqrt[]{m}\right) < 2 \exp(-a^2/2)
 \ee
where $m$ --- is the distance on $\Gamma^{*}$ from $v$ to the nearest boundary face, $m=dist(v, \partial \Gamma)$.
\end{claim}

\begin{proof}

Let us recall that $\mathcal{H}(\Gamma,\chi)$ is a set of height functions on $(\Gamma,\partial \Gamma)$. Consider a probability space $(\mathcal{H}(\Gamma,\chi), \Omega, \mathbb{P})$, where $\Omega:=2^{ \mathcal{\mathcal{H}(\Gamma,\chi)}}$ is $\sigma$-algebra of subsets of $\mathcal{H}(\Gamma,\chi)$, and $\mathbb{P}$ is the Gibbs measure on $\mathcal{H}(\Gamma,\chi)$.

Let $\rho$ be the shortest path on $\Gamma$ from $\partial \Gamma$ to the face $v$ ,  $\rho:= (x_1, x_2 \ldots x_{m})$, where $x_1 \in \pp \Gamma$ and $x_{m}=v$.

We consider a family of equivalence relations on $\mathcal{H}(\Gamma,\chi)$ : $h_1 \sim_{k} h_2 $ if and only if, $h_1(x)=h_2(x)$ for all $x= x_1, \cdots, x_{k} $, so height functions agree at first $k$ points of path $\rho$. It is straightforward that from $h \sim_{k} h^{\prime}$ follows $h \sim_{k-1} h^{\prime}$, thus we may consider a decreasing filtration of $\sigma$-algebras, $\mathcal{F}_k$ in  $\Omega$:

\bb
\mathcal{F}_k:=\{ \text{Equivalence classes under} \sim_{k} \}
\ee

\bb
\mathcal{F}_{k-1} \subset \mathcal{F}_k
\ee

Because boundary conditions are fixed we have $\mathcal{F}_1=\mathcal{H}(\Gamma,\chi)$.
Let us take
the
following sequence of conditional expectation values:
\bb 
M_k:=\mathbb{E} (h(v) \mid \mathcal{F}_k ).
\ee 

$M_k$ is an expectation value of a height function at face $v$ with random values \footnote{Which obey the Gibbs measure.} at first $k-1$ \footnote{We recall that the value at the first face of the path $\rho$ is fixed by the boundary condition and random values start from the second face, thus $M_1$ is an expectation value of $h(v)$. In our parametrization of $\rho=\{x_1, \cdots, x_{k} \}$ where $x_1\in \pp\Gamma$} points of the path $\rho$ (values at
these
points
define
an equivalence class in $\mathcal{F})$. Note that $M_1= \bar h(v)$ because there are no random values \footnote{We do not change $\sigma$-algebra.} and $M_m=h(v)$ because we leave values totally random according to the Gibbs measure. Because $\mathcal{F}_k$ is a filtration, $M_k$ is a martingale:
\bb
\mathbf{E} ( M_{k+1} \mid \mathcal{F}_k )=M_k,
\ee
 which is due to the tower rule: if $\mathcal{G}_1$ and $\mathcal{G}_2 $ are two two $\sigma$-algebras such that $\mathcal{G}_1 \subset \mathcal{G}_2,$ we have
$\mathbb{E}[\mathbb{E}[X \mid \mathcal{G}_2]\mid \mathcal{G}_1] = \mathbb{E}[X \mid \mathcal{G}_1]$.

Also, let us show that $|M_k - M_{k+1}| \leq 1 $.

We need to take an equivalence class in $\mathcal{F}_k$ and then prove the inequality $|M_k - M_{k+1}| \leq 1 .$

A value of a height function from an equivalence class from 
$\mathcal{F}_k$ 
in $x_{k}$ is fixed by the equivalence class and a value in $x_{k+1}$ is a convex combination of the two possible values: $h(x_k)-1$ and $h(x_k)$ with \footnote{Two possible values could be $h(x_k)-1$ and $h(x_k)$ depending on an orientation.} coefficients proportional to the partition function of the dimer covers with this value at $x_{k+1}$.

Note that expectation values for the equal boundary conditions are equal.

On the other hand a value of a height function from
$\mathcal{F}_{k+1}$
at $x_{k+1}$ is fixed by its equivalence class and is either $h(x_k)-1$ or $h(x_k)$ for the same $h(x_k) \in \ZZ$. So $M_k$ and $M_{k+1}$ are expectation values of height functions with boundary conditions that differ less than $1$.

By \hyperref[auxiliary]{the coupling lemma} it is true for expectations.

Note that $M_k$ is a martingale such that $|M_k-M_{k+1}|\leq 1$. Applying Azuma inequality \cite{Azuma} for $M_k$, we get that

\bb
\mathbb{P} (|M_m-   M_1| > a \cdot \sqrt[]{m}) < 2\exp(-a^2 /2 )
\ee
that is the same as
\bb
\mathbb{P} (|h(v)-   \bar h(v)| > a \cdot \sqrt[]{m}) < 2\exp(-a^2/2 ).
\ee

\end{proof}

\begin{corollary}
Moreover if we take an approximation of a domain with a boundary condition $(\Omega,\chi)$ we get the following bound for normalized height function. Let $\hn$ be a random height function on $\Omega_n$ , then

\begin{enumerate}
\item 

\bb
\mathbb{P} (|\hn(v)-   \avhn(v)| > a) < 2\exp(- \frac{1}{2} a^2 n )
\label{ineq}
\ee

\item
\bb
\mathbb{P}(\hn \notin B_{\delta}(\avhn))<K n^2 \exp( -\frac{1}{2} \delta^2 n))
\ee \label{prob}
for some $K>0$ and sufficiently small $\delta$.
\end{enumerate}

\label{concentration}
\end{corollary}
\begin{proof}

Dividing by $n$ inequality in the \hyperref[Concentration_lemma]{Concentration lemma} we get that
\bb
\mathbb{P} (|\hn(v)-   \avhn(v)| > n^{-1}a \cdot \sqrt[]{m}) < 2\exp(-a^2/2 )
\ee

After some rescaling we get the following expression from \ref{ineq},
\footnote{
We can bound $n^{-1/2} \sqrt[]{m}$ by $\sqrt[]{Diam(\Omega)+o(n)}$ because the length of a path is less then the maximal length of a path that is bounded by the diameter of the domain $\Omega$ (in Euclidean metric) and $n^{-1/2}$ is adsorbed due to rescaling of the domain (length of a path between fixed points in $\Omega$ is proportional to $n$ after rescaling). Then we define new constant $c:=a/2 \times n^{-1/2} \sqrt[]{m} \times Diam(\Omega)$
}
\bb
\mathbb{P} (|\hn(v)-   \avhn(v)| > c) < 2\exp(-  c^2 n ). \label{conc}\ee

Then to obtain probability that $\hn \notin B_\delta(g)$ we find $\mathcal{P}(  \hn \in B_\delta(g))$. To do it we need to multiply probabilities that $|\hn(v)-\avhn(v)| \leq \delta $ for all faces $v$ of graph $\Gamma$:

\bb
\mathbb{P}(\hn\in B_{\delta}(\avhn))=\prod_{v \ - \ face \  of\ \Gamma}\mathbb{P}(|\hn(v)-\avhn(v)|\leq \delta)
\ee
Then each factor is equal to $1-\mathbb{P}(|\hn(v)-\avhn(v)|> \delta)$, that we can estimate by using the equality $1-\mathbb{P}=\exp(\log (1-\mathbb{P}))$ and  Taylor expanding of $\log$.
Let us denote $\mathbb{P}(|\hn(v)-\avhn(v)|> \delta)$ by $p$.
\footnote{
$ \log(1-p)=-p+o(p^2)$, and the bound for the factor is $\exp(\log(1-p))=\exp(-p+o(p^2)).$
}
Then the bound
of
the whole product is $\exp((-p+o(p^2) )\times n^2 K)$ where we use the fact that in the domain 
there are
around $K n^2$ vertices for some $K>0.$

Finally we bound $-p$ by $-2\exp(- \delta^{2} n/2)$ using \ref{conc} and get the following 
expression for the probability that a random normalized height function lies in $\delta$-ball around the average normalized height function.
\bb
\mathbb{P}\left(\hn\in B_{\delta}(\avhn)\right)=\exp(-K n^2 \exp(-\delta^2 n/2))).
\ee
It is clear that it converges to one. After using the inequality $1-x\leq\exp(-x)$ and the small change of variables we get $\ref{prob}$.\footnote{Here we have the following change of variables  $\mathbb{P}\left(\hn\in B_{\delta}(\avhn)\right)=1-\mathbb{P}\left(\hn\notin B_{\delta}(\avhn)\right)$,so we can apply the inequality $\exp(-Kpn^2)\ge 1-Kpn^2=1-\mathbb{P}\left(\hn\notin B_{\delta}(\avhn)\right)$ }
\end{proof}

\subsection{Coupling lemma}
\begin{claim} \label{auxiliary}
Let $(\Gamma, \partial \Gamma)$ be a graph with a boundary and let $f$ and $g$ be two boundary height functions on $\Gamma$. We denote by $\bar h_g$ an average height function on $(\Gamma, g)$, $\bar h_f$ on $(\Gamma, f)$.
 We Suppose that $f \leq g$, then $\bar h_f \leq \bar h_g$ pointwise. 

\end{claim}

\begin{proof}

Let $\mathcal{H}(\Gamma,g)$ and $\mathcal{H}(\Gamma,f)$ be sets of height functions on $\Gamma$ with boundary height functions $f$ and $g$. Denote induced Gibbs measures by $\mu_f$, $\mu_g$.

Let us prove that $f \leq g$ $\Rightarrow$ $\bar h_f  \leq \bar h_g $.
To do this we need to build a \textit{coupling} of measures $\mu_f$ and $\mu_g$. 
It is a probability measure $\pi$ on $H(\Gamma,g) \times H(\Gamma,f)$, such that its projection on $H(\Gamma,g)$ gives $\mu_g$ and the same is true respectively  for $f$. The most important constraint is that $\mathbb{P}_{\pi}\left( h_1,h_2 \right) =1$ for only pairs of height functions such that $h_1 \leq h_2$. In case we have such a measure, we get $\bar h_f \leq \bar h_g$ pointwise, because $\bar h_f$ and $\bar h_g$ can be computed using $\pi$ due to the fact that each measure can be obtained as a projection of $\pi$.

Let us prove it using induction by number of internal vertices of $\Gamma$.
Base of induction where the set of internal vertices is empty is trivial.
Moreover the case of $f=g$ is trivial\footnote{in case of equal boundary conditions expectations are just the same.} and it is sufficient to prove for $f<g$ at some boundary face $v$. Note that $h_f(v)\leq h_g(v)+1$. Let us pick an internal face $w$ adjacent to $v$. There are two possible values for $h_g(w)$ - $h$ and $h-1$ for some $h$ ( respectively  $h^{\prime}$ and $h^{\prime}-1$ for $h_f(w)$ for some $h^{\prime}$).  
\bb
\mu_f=\alpha \mu_f^{+}+(1-\alpha) \mu_f^{-}
\ee
\bb
\mu_g=\beta \mu_g^{+}+(1-\beta) \mu_g^{-}
\ee
where $\alpha$ is proportional to the partition function of the dimer covers with the extra boundary condition $f^{+}$ and $\beta$ is proportional to the partition function with the extra boundary condition $f^{-}$ at $w$. And $\mu_f^{\pm}$ (respectively $g$) is the Gibbs measure on height functions on $(\Gamma,\partial \Gamma)$ with an extra boundary condition at $w$. And thus we can use our induction hypothesis to conclude that the coupling exists for any pair measures with an extra boundary condition at $w$.

Then let us consider the case of $\beta=1$. In this case we have two couplings $\pi^{+}$ and $\pi^{-}$ between $\mu_f^{+}, \mu_f^{-}$ and $\mu_g^{+}$. We can take the coupling to be a superposition: $\pi=\alpha\pi^{+}+(1-\alpha)\pi^{-}$.
\bb
\mu_f=\alpha \mu_f^{+}+(1-\alpha) \mu_f^{-}
\ee
\bb
\mu_g=\beta \mu_g^{+}
\ee

The case of arbitrary $\beta$ is just a consequence of the previous one. 

\bb
\mu_f
\ee
\bb
\mu_g=\beta \mu_g^{+}+(1-\beta) \mu_g^{-}
\ee
We know from the previous case that there exist couplings of $\mu_f$ with both $\mu_g^{+}$ and $\mu_g^{-}$. Then we can just take a superposition of
 $\mu_f$ with the $\mu_g^{+}$ and $\mu_g^{-}$ with the weights $\beta$ and $1-\beta$.

\end{proof}

\section{Proofs of properties of height functions}
\label{8}
Here we formulate and prove several propositions about height functions, all these statements are quite the same and the simplest case is $1$-Lipschitz functions where the structure is easier to understand and that we prove in the next section\footnote{1-Lipschitz functions are asymptotic height functions for the case of $N_G=\mathbb{D}^1$.}. 

\subsection{The main construction for $1$-Lipschitz functions}

Let $(\Omega,\chi)$ be a domain with a boundary condition. One can ask a question, under which constraints on $\chi$ it admits a continuation to a $1$-Lipschitz function on $\Omega$.
\begin{proposition}
Suppose that we have a domain with a boundary condition $(\Omega, \chi)$.
Then $\chi$ admits a continuation to a $1$-Lipschitz function if the following inequality holds,
\bb
|\chi(x)-\chi(y)|\leq|x-y|, \forall x,y \in \Omega.
\label{Lipschitz1}
\ee
\end{proposition}
\begin{proof}
Due to the fact that the pointwise minimum of $1$-Lipschitz functions is a $1$-Lipschitz function, we can construct an extension by the following formula,
\bb
h(x):=\min_{y\in \ \pp\Omega}{\chi(x)+|x-y|}
\ee

We know that $h$ is a $1$-Lipschitz function and we need to show that $h$ coincides with $\chi$ on $\pp \Omega$.
To do this we can prove that $\chi(x) \leq h(x) \leq \chi(x),\forall x \ \in\pp\Omega $. 

The first inequality archives the equality for the function from the family correspondents to the point $y=x$.

The second inequality $\chi(x) \leq h(x)$ holds for all functions among which we take the minimum,
\bb
\chi(x)\leq \chi(y)+|x-y|
\ee
It is just the Lipschitz condition \ref{Lipschitz1}. 
\end{proof}
In following sections we repeat almost the same propositions thrice.

\subsection{Piecewise linear approximations of asymptotic height functions}
In this subsection we recall piecewise linear approximations of Lipschitz functions that we use in the proofs.

Let us take $\ell > 0$ and take a triangular lattice with equilateral triangles of side $\ell$. 
We map an asymptotic height function $h \in \mathscr{H}(\Omega)$ to a piecewise linear approximation, that is linear on every triangle, moreover it is the unique linear function that agrees with $h$ at vertices of the triangle. Let us  denote this approximation of $h$ by $\hat h$.
Then the following statements are true by the Lemma 2.2,from \cite{CKP},

\begin{claim} \label{analitic1} 
Let $h \in \mathscr{H}(\Omega)$ be asymptotic height function and let $\epsilon>0$. Then for sufficiently small $\ell>0$, on at least $1-\epsilon$ fraction of the triangles in the $\ell$-mesh of triangles that intersect $\Omega$ we have the following property: $\norm{h-\hat h}\leq \ell \epsilon$.
\end{claim}

\subsection{The density lemma and the criteria of existence.}
Here we prove a sequence of auxiliary propositions that lead us to the Density Lemma and the Criteria of existence.

\subsubsection{Criteria for asymptotic height functions.}
Before the propositions we need to introduce the support function of $N_G$. It is a function $\theta: \RR^2 \to \RR^2$ defined by the following formula,

\bb
\theta(x):=\min\left(\lambda \geq 0 \ \ | \ \ \frac{x}{\lambda} \in N_G\right) .
\ee
It is useful to keep in mind the case of 1-Lipschitz functions, where $N_G$ is just a unit disk and $\theta(x)=|x|$.
The support function has a value $1$ on $\pp N_G$ and its slope lays on $\pp N_G$, so $\theta$ is an asymptotic height function. Its value at $x=(0,0)$ is $0$ and it is the maximal asymptotic height function among such asymptotic height functions that have this property in the sense that for any such asymptotic height function $\phi$ the following inequality holds, 
\bb
\phi(x)\leq \theta(x,) \forall x \in \RR^2.
\ee
Note that we can rewrite it for any asymptotic height function $\psi$ substituting $\phi(x)=\psi(x)-\psi(y)$

\bb
\psi(x)-\psi(y)\leq\theta(x,y), \forall x,y \in \RR^2
\ee

Let us call it the \textit{Support Lipschitz condition} \footnote{which becomes usual $1$-Lipschitz condition for $N_G=\mathbb{D}^1)$}.

Suppose that we have a domain with a boundary condition, $(\Omega,\chi)$. One can ask a question, under which constraints on $\chi$ it has a continuation to an \textit{asymptotic height function} on $\Omega$.

\begin{proposition}
Let $(\Omega,\chi)$ be a domain with a boundary condition.
Then, $\chi$ admits a continuation to an asymptotic height function if it satisfies the following inequality,
\bb
\forall x,y \in \pp\Omega  \ |\chi(x)-\chi(y)|\leq \theta(x-y)
\label{Lipschitz}
\ee
\
\end{proposition}

\begin{proof}
Let us take a family of asymptotic height functions, $\{\chi(y)+\theta(x-y) \}, x\in \ \Omega$ and $y \in \ \pp \Omega$. Then, let us take the pointwise minimum over this family.

\bb
h_{\max}(x):=\min_{y\in \pp\Omega}(\chi(y)+\theta(x-y))
\ee
$h_{\max}$ is an asymptotic height function because the pointwise minimum of asymptotic height functions is again an asymptotic height function. So we need to show that $h_{\max}$ agrees with $\chi$ on $\pp\OO$.

To do this let us prove that $\chi \leq h_{\max}\leq \chi$ on the boundary of the domain.
Let $x$ be an arbitrary point $x \in \pp\OO$, the first inequality 
\bb
\chi(x)\leq \chi(y)+\theta(x-y).
\ee

That is just (\ref{Lipschitz}).
The second inequality becomes an equality for the asymptotic height function from the family correspondents to the point $x=y$.
\end{proof}
Due to (\ref{Lipschitz}) $h_{\max}$ is the maximal extension of $\chi$ to an asymptotic height function $\OO$. The minimal extension, $h_{\min}$ has the similar form,
\bb
h_{\min}(x):=\max_{y\in \pp\Omega}(\chi(y)-\theta(x-y)).
\ee

\subsection{The criteria for height functions.}
Let $(\Gamma,\chi)$ be a graph with a boundary condition.
One can a question under which assumptions on $\chi$ it extends to a height function on $\Gamma$.

Before formulating the criteria we need to introduce the support height function that is a height function defined on the whole $G$.

Let $x \in G$ and let us look at all height functions that have the value $0$ at face $x$ and let us take the pointwise maximum among such height functions. Note that a height function of periodic dimer covers of $G$ are defined up to an additive constant that we can pick so that a value at $x$ is $0$.

\bb
\hat\theta(x,y):= \max \{ h(x) \ | \ h\textrm{ is a height function such that } h(y)=0 \}.
\ee

\begin{proposition}
Let $\GB$ be a graph with a boundary. Then $\chi$ admits an extension to a height function on $\Gamma$ if the following inequality holds,

\bb
\chi(x)-\chi(y) \leq \hat \theta(x,y)
\label{Height_Lipschitz}
\ee
\end{proposition}

\begin{proof}
Let us take a family of height functions, $\{\chi(y)+\hat\theta(x,y) \} $ and let us take the pointwise minimum over this family.

\bb
h_{\max}(x):=\min_{y \in \pp\Gamma}(\chi(y)+\hat\theta(x,y))
\ee
$h_{\max}$ is a height function because the pointwise minimum of height functions is again a height function. So we need to show that $h_{\max}$ agrees with $\chi$ on $\pp\Gamma$.

To do it let us prove that $\chi \leq h_{\max}\leq \chi$ on the boundary faces.
Let $x$ be an arbitrary boundary face, $x \in \pp\OO$. It is sufficient to show that there exists a height function from the family such that it satisfies the first inequality for the continuation that is the following,

\bb
\chi(x)\leq \chi(y)+\hat\theta(x,y).
\ee
It is the statement that the height function $\chi(x)-\chi(y)$ that has a value $0$ at $y$ is less than the maximal height function with this property,
\bb
\chi(x)-\chi(y) \leq +\hat\theta(x,y).
\ee

The second inequality becomes an equality for the point $x=y$.

\end{proof}
Note that due to (\ref{Height_Lipschitz}) $h_{\max}$ is the maximal extension of $\chi$ to a height function $\Gamma$. The minimal extension of $\chi$ can be constructed by almost the same way as the maximal. Let us define the minimal extension $h_{min}$,
\bb
h_{\min}(x):=\max_{y \in \pp\Gamma}(\chi(y)-\hat\theta(x,y)).
\ee

\subsection{Convergence of the maximal extensions}
Notice that the strategy of proof of the criteria is almost the same. The key point is the use of the support functions.  To understand a convergence of maximal extensions we need to understand a convergence of the support functions.
Before the proposition we need to normalize the support height function,
\bb
\theta_n:=\frac{\hat\theta}{n}.
\ee

\begin{proposition}
Let us consider a support height function $\theta_n(,y)$ and the asymptotic height function $\theta(,y)$, where a face $y$ contains $(0,0)$\footnote{ So $\theta_n(,y)$ is defined at $(0,0)$ with the value $0$.}.
Then there is the following convergence,
\bb
\theta_n(,y) \to \theta(,y) \textsf{ as } n\to \infty,
\ee

\end{proposition}

\begin{proof}
By definition, $\theta$ and $\theta_n$ have the same values at $0$. And to prove a convergence we need to understand that the slope of $\theta_n$ converges to the slope of $\theta$.

Let us take an arbitrary direction in $\RR^2$, the slope of $\theta$ along the direction defines a point on $\pp N_G$, denote the point $(s,t)$.

From Proposition 3.2 in \cite{KOS} we know that there exists a sequence of dimer covers of $G$ with the normalized slopes converge to $(s,t)$ as $n\to \infty$, denote the sequence $D_n$.

By the definition of $\theta_n$ its slope is not less than the slope of $D_n$. So the slope of $\theta_n$ converges to the slope $\theta$.
\end{proof}

Now we are ready to prove a convergence of the maximal extensions.

\begin{proposition}
Suppose that $(\OO,\chi)$ is a domain with a boundary condition, and $\Gn$ its approximation. Let $\eta_n^{\max}$ be the maximal extension of $\chi_n$.
Then 
\bb
\eta_n^{\max} \to h_{\max} \textbf{ as } n \to \infty
\ee
\end{proposition}

\begin{proof}
Let us look at expressions of the maximal extensions,
\bb
h_{\max}(x):=\min_{y\in \pp\Omega}(\chi(y)+\theta(x,y))
\eta_n^{\max}(x):=\min_{y \in \pp\Gamma}(\chi_n(y)+\theta_n(x,y)).
\ee
To prove the convergence we need to show that the pointwise minimum of $\chi_n(y)+\theta_n(x,y)$ converges to the pointwise minimum of $\chi(p)+\theta(x-p)$. Because $\Gn$ is an approximation, $\chi_n(y)$ converges to $\chi(y)$ and $\theta_n(x,y)$ converges to $\theta(x-y)$. So if $\chi(p)+\theta(x-p)$ archives its minimum at some point $p$, $\chi_n(y)+\theta_n(x,y)$ must archive its minimum at a face that contains $p$\footnote{On the other points of the face, a value of $\chi+\theta$ is $o(n^{-1})$-close to $\chi(p)+\theta(x-p)$ which is $o(n^{-1})$-close to $\chi_n(y)+\theta_n(x,y)$.}.  Otherwise it will converge to $\chi(z)+\theta(x,z)$ for some point $z$, that is bigger then a value of $\chi+\theta$ at point $p$.
\end{proof}

\subsection{Density lemma}
Now we are ready to prove the density lemma.
Before the formulation and its proof note that for a domain with a boundary condition $(\Omega, \chi)$ there is a tautological way to express an asymptotic height function $f$ in terms of its own values on $\Omega$.
\bb
\hat f(x):=\min_{y \in \ \Omega}(f(y)+\theta(x,y))
\label{support123}
\ee

Again we need to show that that $\hat f \leq f \leq \hat f$,
first inequality $f\leq f(y)+\theta(x,y)$ is just the consequence of the Support-Lipschitz condition. And the equality for the second inequality is archived at for the function correspondents to $y=x$.

\begin{theorem}[The Density lemma]
Let $( \Omega_{n}, \chi_{n} )$ be an approximation of $(\Omega, \chi)$.
Then for every $h \in \mathcal{H}(\Omega,\chi) $ there exists a sequence of normalized height functions $\eta_n$, such that $\norm{h-\eta_{n} }\leq \frac{C}{n}$, where $C$ is a positive constant. 

Vice versa, for every normalized height function $\eta_{n}$ on $(\Omega_{n},\chi_{n})$ there exists an asymptotic height function $h \in \mathcal{H}(\Omega,\chi)$ such that $\norm{h-\eta_{n}} \leq n^{-1}C$.
\label{density}
\end{theorem}

\begin{proof}
The strategy is to use \hyperref[support123]{the tautological expression} for function $h$ and its analogue for height functions.

Let us take a family of points $F:=\{p_f\}$ on $\Omega$, where $f$ is a face on $\Omega_n$ and for every face $f$ there exists a point from the family that contains exactly one point from the family. 

Let us define the height function that approximates $f$.
\bb
\hat\eta_{n}(x):=\min_{y \in \{p_f \}_{f \in \Omega}}(f(y)+\theta_n(x,y)),
\ee
where $f(y)$ is a value of $f$ at a point from the family $\{p_f\}$ corresponding to the face $y$. 
We claim that $\hat\eta_n$ approximates $f$. The proof is almost the same as for the maximal extension. Let us show that $\hat\eta_n(x) \leq f(x) \leq \hat\eta_n(x)$.

Notice that we have to show it up to an error of order $o(n^{-1})$\footnote{We can add a term $C n^{-1}$ for sufficiently large $C$ to make the first inequality true, while the limit as $n\to \infty$ will be the same.}.

Note that $f(y)+\theta_n(x,y)$ converges to $f(y)+\theta(x,y)$ as $n \to \infty$. Thus the second inequality becomes the Support Lipschitz condition.

The equality for the first inequality is achieved at the face containing point $x$. If $x \in F$ the equality is exact, otherwise it is true up to an error of order $n^{-1}$ which is $f(x)-f(x^{\prime})$ where $x^{\prime}$ is a point from the family correspondent to the face $x$.

It possible that $\hat\eta_{n}$ will have a wrong boundary condition, so we need to balance it between the maximal and the minimal height functions on $\Omega_n$. 
\bb
\eta_{n}:=\max(h_{max}^n,(\min(h_{\min}, \hat \eta_{n})))
\ee

The proof of the second part of the statement is almost the same,
let us define an asymptotic height function that approximates $\hn$.
\bb
f(x):=\min_{y \in \{p_f \}_{f \in \Omega}}\hn(x)+\theta(x,y)).
\ee
\end{proof}

\section{The Surface tension theorem.}
\label{9}
Finally, we can formulate the main auxiliary theorem of the proof.

\begin{theorem}
Let $\Omega$ be a compact, connected, simply-connected domain in $\mathbb{R}^2$. Then let $(\Omega.\chi)$ be the domain with a boundary condition $g \in \mathscr{H}(\Omega,\chi)$  and $A_{n}=(\Omega_{n},\chi_{n})$ be an approximation of $(\Omega,\chi)$. We recall that 
$Z(g, \delta \, | \, A_{n})$ is a partition function of configurations with height functions $\delta$-close to $g$. Then

\begin{equation}
\lim_{ \delta \to 0}\lim_{n \to \infty}{ n^{-2} \log Z(g,\delta \ | A_{n})}= \int_{\Omega}{\sigma(\nabla g) dx dy}
\end{equation} 
\label{Surface}
\end{theorem}

Before the proof we need two lemmas, \textit{the Square lemma} and \textit{the Triangle lemma}. 
Before the Square lemma we need to prove an auxiliary proposition.

\subsection{The Slide lemma}
 Suppose that we have a unit square $\cS$ and a smaller square of length $1-2\ell$ inside $\cS$ with the same center( see Figure \ref{squaresquare}). 
Let $\psi$ be a linear function on the smaller square and let $\chi$ be an asymptotic height function that $\delta$-close to a linear function with the same slope as $\psi$.

\begin{proposition}[The slide lemma]
In the above notations there is an extension of $\psi$ to an asymptotic height function on $(\cS,\chi)$ for sufficiently small $\delta$ and a proper choice of $\ell$.
\label{slide}
\end{proposition}

\begin{figure}[ht]
\centering
\includegraphics[width=0.45\textwidth]{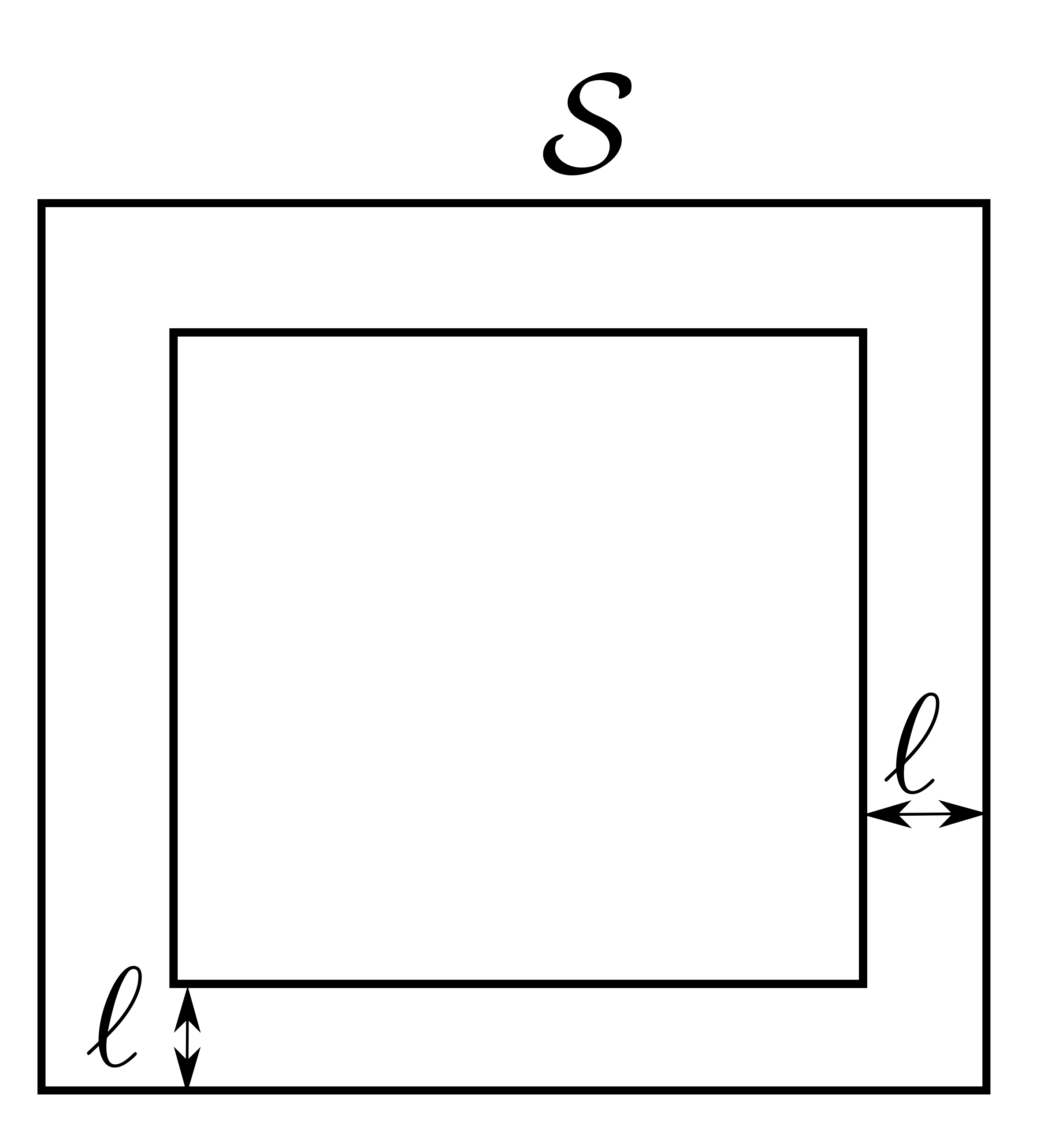}
\caption{A unit square containing a square of size $1-2\ell$}
\label{squaresquare}
\end{figure}

\begin{proof}
We need to extend $\psi$ to an asymptotic height function $\hat \psi$ that coincides with $\chi$ on $\pp\cS$. We will construct a Lipschitz extension and then choose parameters to make it an asymptotic height function.
First we make a shift of all functions by the linear function with the same slope as $\psi$. Then we have zero on the smaller square and the boundary condition on $\cS$ bounded in modulus by $\delta$, without loss of generality let us denote the shifted boundary condition by the same letter $\chi$. 

Let us connect two values on every ray by the unique linear function that has a value $0$ on the boundary of the smaller square and a value of $\chi$ on $\pp\cS$. Let us denote the continuation by $h$.
\begin{figure}[ht]
\centering
\includegraphics[width=0.45\textwidth]{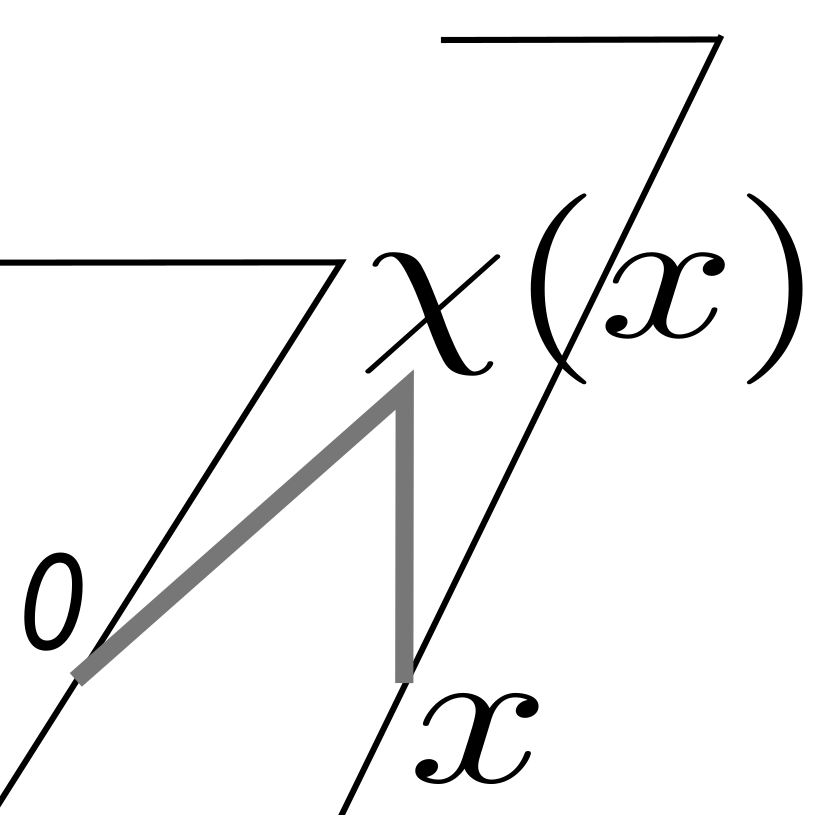}
\caption{A continuation}
\label{squarehill}
\end{figure}

We have to prove that $h$ is an asymptotic height function. To do this we need to check that $\nabla h(x) \in N_G$. 

On each ray $\nabla h(x)$ is a constant equal to $\frac{\chi(x)}{\ell}$ and in direction orthogonal to the ray, $h$ has the same slope as $\chi$.
Let us choose $\ell=\frac{\sqrt[]{\delta}}{4}$, then for sufficiently small $\delta$ the slope after inverse shift lies on $N_G$.
\footnote{Note that not all boundary height functions $\chi$ are possible to appear. For example in case of $(s,t) \in \pp N_G$ and $s,t \geq 0$ height functions on the north boundary are less than zero and on the south are more than zero due to the fact that the slope before the linear shift lies in $N_G$.}
\end{proof}
\begin{remark}
Note that the area of continuation is of order \ $\sqrt[]{\delta}$.
\end{remark}

\subsection{The square lemma}
\begin{proposition}[The square lemma]
Let $(\cS,\chi)$ be a domain with a boundary condition, where $\cS$ is a unit square and $\chi$ is $\delta$-close to a linear function with the slope $(s,t)$.

Suppose that $(\cS_{n},\chi_{n})$ is an approximation of $(\cS,\chi)$ and let $g \in \mathscr{H}(\cS,\chi)$ be $\delta$-close to a linear with the slope $(s,t)$. We recall that $Z(g,\delta \ | \cS_n)$ is the partition function of $\cS_n$ where we take into account only configurations with height functions $\delta$-close to $g$.
Then,
\label{square_case}
\begin{equation}
\lim_{ \delta \to 0}\lim_{n \to \infty}{ n^{-2} \log Z(g,\delta \ | \cS_n)}= \sigma(s,t)
\end{equation}
\end{proposition}
In this proposition all graphs will be subgraphs of $G_n$ for a fixed $n>0$ which we hold until the end of the proposition and then send it to infinity.
Before the proof we need to introduce two sequences of graphs embedded into the torus.
\bb
T_{n\pm l(n)}:=\phi_{n}(G/(n\pm l(n))\ZZ^2)
\ee
Where $l(n)$ is a sequence such that $\lim_{n \to \infty}\frac{l(n)}{n}=\ell$.

\begin{proof}
We want to approximate the partition function of $\cS_n$ by the partition functions of $T_{n\pm l(n)}$ that include dimer cover of the slope $(s,t)$ (let us denote the partition function of configuration with this slope by $Z_{(s,t)}$). The weights of the dimer covers are the same, so we have to show that the number of such dimer covers is approximately the same. Let us recall that $\mathfrak{D}_{s,t}(T_{n- l(n)})$ is the set of dimer configurations on $T_{n- l(n)}$ that have the slope $(\lfloor{ns}\rfloor,\lfloor {nt}\rfloor)$. We want to show that there are two inclusions,
\bb
\mathfrak{D}_{s,t}(T_{n- l(n)})\subseteq \mathfrak{D}(\cS_{n}|\delta)\subseteq \mathfrak{D}_{s,t}(T_{n+ l(n)}).
\label{inclusion}
\ee
Then we can conclude the following relation between the partition functions\footnote{Note that the weight of dimer covers and the weight of their extension are not the same. However, the dimer covers differ only on the area of continuation that is of order $\delta$. Thus the weights differ by $C^{\text{Area}} =C^{\prime}\exp(o(\delta))$ },
\bb
\sum_{D\in \mathfrak{D}_{s,t}(T_{n-l(n)})} w(D)\leq \sum_{D\in \mathfrak{D}(\cS_{n}|\delta)} w(D)\leq \sum_{D\in \mathfrak{D}_{s,t}(T_{n+l(n)})} w(D).
\ee
And after the normalization,
\bb
n^{-2}\log(Z_{s,t}(T_{n-l(n)}))\leq n^{-2}\log( Z(\cS_{n}|\delta))\leq n^{-2}\log (Z_{s,t}(T_{n+l(n)}).
\label{normalization}
\ee
We know that $\lim_{n \to \infty}n^{-2}\log(Z^{(s,t)}(T_n)=\sigma(s,t)$ from \ref{surface3}.

So we can conclude that 
\bb
\lim_{n\to \infty}(n\pm l(n))^2 \times(n\pm l(n))^{-2}  n^{-2}\log(Z_{s,t}(T_{n\pm l(n)}))=\sigma(s,t) \times \lim_{n \to \infty}\frac{(n\pm l(n))^2}{n^2}=\sigma(s,t)(1+o(\delta^{1/2}).
\label{normalization_tori}
\ee

Let us prove the first inclusion in \hyperref[inclusion]{(\ref{inclusion})} , to do it we have to show that each height function on $(T^{s,t}_{n- l(n)})$ extends to a height function on $(\cS_n|\delta,g)$.
Note that $T^{s,t}_{n-l(n)}$ is an approximation of the smaller square inside $\cS$ (let us denote the smaller square by $\mathfrak{S}$ ) with the linear boundary condition $\chi$ that has the slope $(s,t)$.
First let us choose the extension of the boundary condition on $\pp \mathfrak{S}$ to an asymptotic height function on $\cS$.
From the previous proposition we know that there is a choice of $\ell$ such that the extension of the boundary condition to an asymptotic height function on $\cS$ exists.

Then using the density lemma we can build a sequence of height functions on $(\cS_n,\chi_n)$ such that on $T_{n-l(n)}$ the sequence coincides with the approximatively linear height function.

To prove the second inclusion we have to show that each height function on $(\cS_n|\delta,g)$ extends to a height function on $(T^{s,t}_{n+l(n)})$. The proof is almost the same as for the first inclusion except the way of using the previous proposition. We have to invert \hyperref[slide]{the slide lemma} to the case where the boundary condition on the boundary of the bigger square is linear with the slope $(s,t)$ and the boundary condition on the boundary of the smaller square is bounded in modulus by $\delta$ \footnote{The extension is the same, we make the shift by the linear function with the slope $(s,t)$ and on each ray between boundaries we connect two boundary conditions by the unique linear function that coincides with the boundary conditions. }.
\end{proof}
\begin{remark}
\hyperref[square_case]{The Square lemma} also works for tori of different sizes.
Let us look at the square of size $k\times k$ with embedded $G_n$,let us denote the image of the embedding by $S_n(k)$. Doing the same as for \hyperref[normalization_tori]{the normalization of tori} we get an area factor,
\bb
\lim_{n\to \infty} (kn)^{2} \times(kn)^{-2}  n^{-2}\log(Z^{(s,t)}(S_n(k))=\sigma(s,t) \times k^2.
\ee
\end{remark}

\subsection{The triangle lemma.}

\begin{claim}
Let $\Omega$ be an equilateral triangle in $\mathbb{R}^2$ and $\chi \in \mathscr{H}(\Omega)$. Let $g \in \mathscr{H}(\Omega,\chi)$ be an asymptotic height function $\delta$-close to a linear with the slope $(s,t)$ and let $A_{n}=(\Omega_{n},\chi_{n})$ be an approximation of $(\Omega,\chi)$.

\begin{equation} \label{Functional-triangle}
\lim_{ \delta \to 0}\lim_{n \to \infty}{ n^{-2} \log Z(g,\delta \ | A_{n})}= \sigma(s,t)
\end{equation} 

\label{triangle_case}

\end{claim}

\begin{proof}
 Let us take a square lattice with the mesh $\ell$ \footnote{Where we take $\ell \leq \delta$ such that $\ell=o(\delta)$} and use the cutting rule for the curve $\rho$ obtained from the intersection of the square lattice and $\Omega$.

\bb
Z(\Omega_n,\chi_n)=\sum_{\chi_\rho}{Z(\chi_{\rho})}
=\sum_{\chi_\rho} \prod_{j}{Z(\cS^{j},\chi_{\rho}^{j})},
\label{cutting_sum_square}
\ee
where $\cS_j$ is the square from the Square lattice with the boundary height function $\chi_{\rho}^{j}$.

We have two types of squares: included squares that do not intersect $\pp(\Omega)$ and excluded squares that intersect $\pp(\Omega)$. We want to build an upper and a lower bounds for the partition function of $(\Omega_n,\chi_n|g,\delta)$, let us denote the upper bound by $Z_L(\Omega_n)$ and the lower bound by $Z_U(\Omega_n)$.
We estimate included squares using the Square lemma and make a rough bound for the excluded squares. 
\begin{figure}[ht]
\centering
\includegraphics[width=0.45\textwidth]{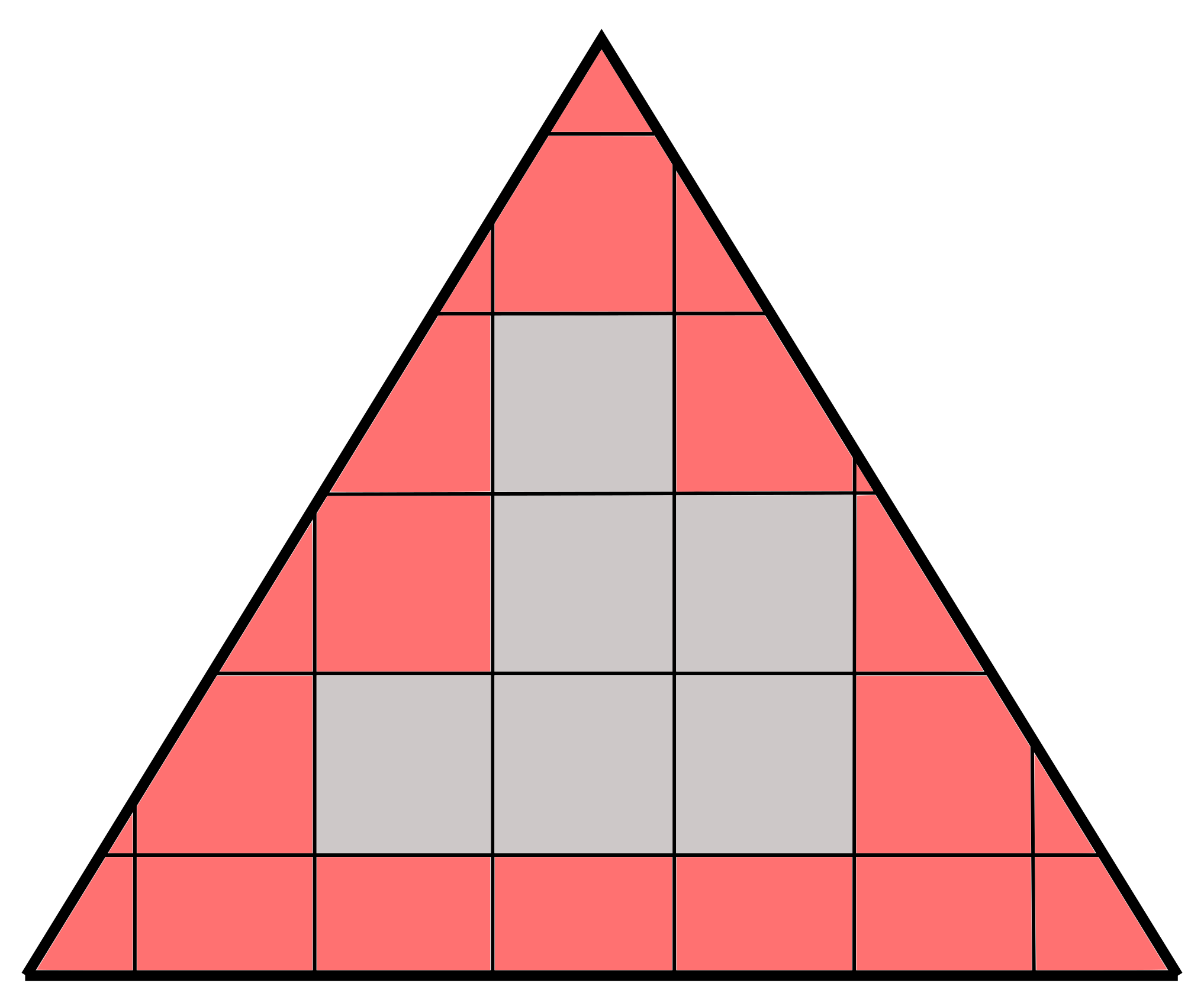}
\caption{An intersection of a triangle with a square lattice. Included squares in gray and excluded ones in red.}
\end{figure}

\subsubsection{The lower bound.}
We want to make a lower bound of sum from \hyperref[cutting_sum_square]{the cutting rule}, let us take only one summand correspondent to a boundary height function $\chi_n$.\footnote{Existence of such height function follows from the density lemma applied to $g$. A boundary height function is a restriction of the height function that approximates $g$.}

Here we estimate the excluded squares by the minimal weight of an edge (let us denote it by $w_{\min})$ to the power of number of edges in these squares.

\bb
\log Z(\Omega_n,\chi_n) \geq  \log Z_L(\Omega_n):= \log{w_{\min}}S_{ex} + \sum_{k}\log{Z_n^k(s,t)},
\ee
where $S_{ex}$ is the total number of edges in the excluded squares that we estimate by $o(\ell)$. And $Z_n^k(s,t)$ is the partition function of the square $S^k$.
After the normalization we have $\lim_{n \to \infty} {n^{-2}\log{Z_n^k(s,t)}}= \sigma(s,t)\times\mathcal{A}(S_k)+o(\delta)$ from the square lemma and remark 1. 
Finally, the lower bound is,
\bb
n^{-2}\log Z_L(\Omega_n)=\sum_k{n^{-2}\log{Z_n^k(s,t)}}+o(\delta).
\ee

\subsubsection{The upper bound.}
In the upper bound we estimate the sum in \hyperref[cutting_sum_square]{the Cutting Rule} by taking the maximal summand times the number of summands, which is the number of the all boundary height functions.

We bound the excluded squares by the maximal weight of the edge to the power of total number of edges. For the included squares we estimate the same way as for the lower bound. Let us bound the number of boundary height function by $2^{o(n)}$

\bb
Z(\Omega_n,\chi_n) \leq Z_U(\Omega_n):=w_{\max}^{o(\ell)}\times 2^{o(n)}\times \sum_j{n^{-2}\log{Z_n^j(s,t)}}
\ee

Finally after taking the limit as $n\to \infty$ the first summands differ from each other by $o(\delta)$ . Thus, both estimates differ from $\sigma(s,t)\times \textrm{Area of the triangle}$ by $o(\ell)$. So after taking limit as $\ell \to 0$ and $\delta \to 0$ we have the proposition. 
\end{proof}

\subsection{The proof of the Surface Tension theorem}
In this section we will use a triangular lattice with a mesh $\ell$ and peacewise linear approximations of Lipschitz functions from \hyperref[analitic1]{Claim \ref{analitic1}}. The idea of the proof is almost the same as for the triangular lemma. 

\begin{theorem}
Let $\Omega$ be compact, connected, simply-connected domain in $\mathbb{R}^2$ and $\chi \in \mathscr{H}(\Omega)$. Let $g$ be an asymptotic height function $ g \in \mathscr{H}(\Omega,\chi)$  and let $A_{n}=(\Omega_{n},\chi_{n})$ be an approximation of $(\Omega,\chi)$. Then

\begin{equation} \label{Functional}
\lim_{ \delta \to 0}\lim_{n \to \infty}{ n^{-2} \log Z(g,\delta \ | A_{n})}= \int_{\Omega}{\sigma(\nabla g) dx dy}
\end{equation} 
\end{theorem}

\begin{proof}
From $\hyperref[density]{ the Density \ lemma}$ we have a sequence of normalized height function $\{\hn \}$ that converges to $g$. 

Take $ \ell $ small enough such that $\norm{g_\ell-g}\leq \delta/3$ on $1-\delta /3$ fraction of the triangles\footnote{It corresponds to a choice$\epsilon=\delta/3$ in \hyperref[analitic1]{claim \ref{analitic1}}.}.

Let us use the $\hyperref[Cutting_boundary]{Cutting \ Rule}$ for curve $\rho$ obtained from the intersection of the triangular lattice of the mesh $\ell$ with the domain. Note that the curve cuts $\Omega_n$ into triangles with boundary conditions along $\rho$, let us denote the triangles by $\{T^{j}\}$ and their boundary height functions by $\{\chi_{\rho}^{j}\}$. (See figure \ref{general})

\bb
Z(\Omega_n,\chi_n)=\sum_{\chi_\rho}{Z(\chi_{\rho})}
=\sum_{\chi_\rho} \prod_{j}{Z(T^{j},\chi_{\rho}^{j})},
\label{cutting_sum}
\ee

There are two types of triangles $\{T^j\}$. The first type is triangles that do not intersect the boundary of the domain and where $g_\ell$ is $\delta/3$-close to $g$. The second type consists of triangles that intersect the boundary of the domain or where $g_\ell$ does not approximate $g$.

\begin{figure}[h!]
\centering
\includegraphics[width=0.3\textwidth]{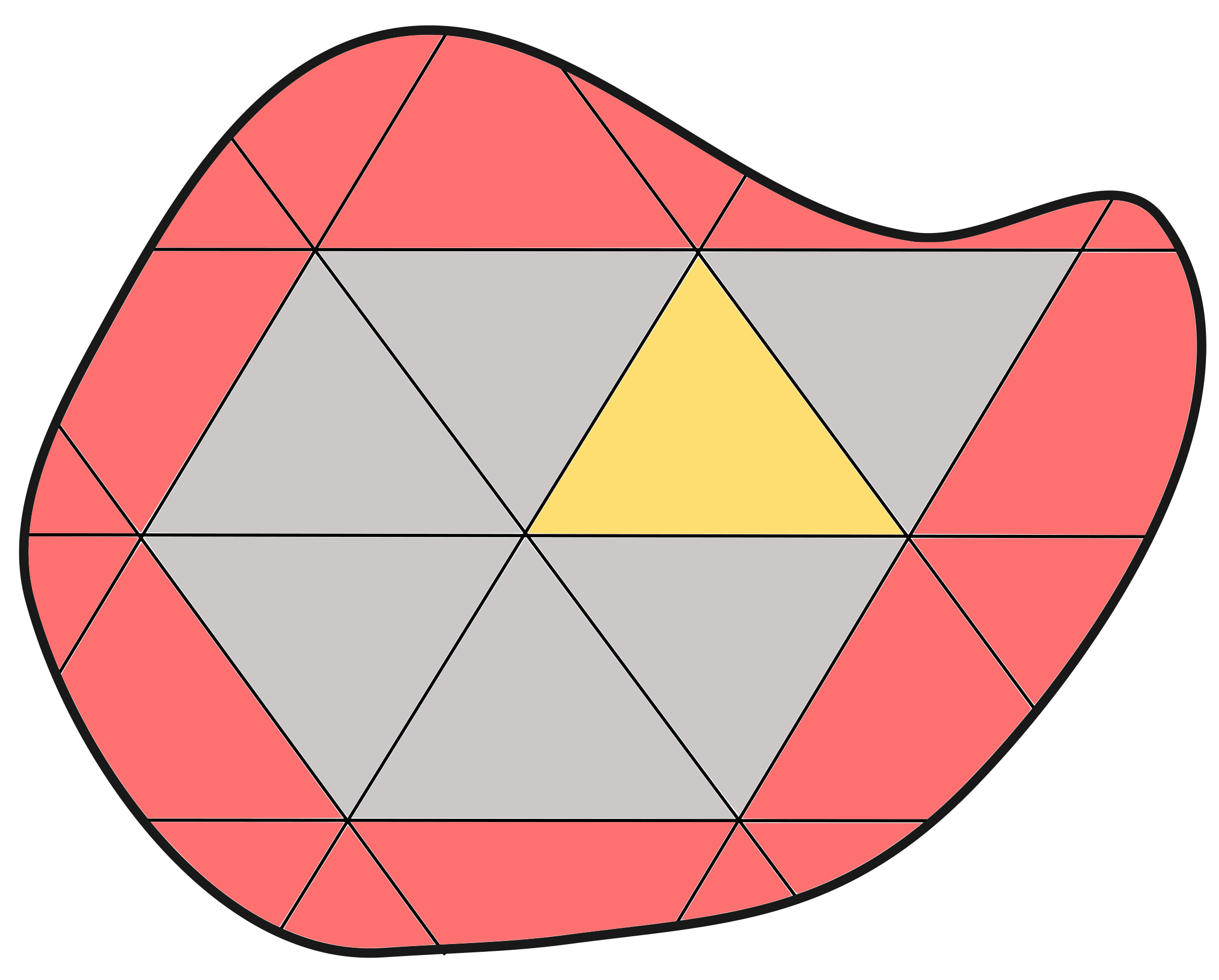}
\caption{Intersection of $\Omega_n$ with a triangular lattice. Triangles of the first type in gray and triangles of the second type in yellow and red.}
\label{general}
\end{figure}

Then we want to make an upper and a lower bounds for the normalized partition function, $n^{-2}{\log Z(\Omega_n,\chi_n|\delta)}$. In both cases we estimate two types of triangles separately. For the first type we can use  \hyperref[triangle_case]{claim \ref{triangle_case}} and for the second we make a rough estimate. Then after taking limit as $n\to \infty$ the normalized estimates will differ by $o(\delta)$.

\subsubsection{The lower bound.}

In a lower bound we have to include some height functions that $\delta$-close to $g$. 
To do this we can take only one summand from \hyperref[cutting_sum]{(\ref{cutting_sum})}  corresponding to one boundary height function $\chi_\rho$. Let us take the boundary height function obtained from the restriction of $\hn$, without loss of generality let us denote it  $\chi_\rho$.

Let us estimate the triangles of the first type by the product that includes only triangles of this type. We can bound the triangles of the second type by the weight of one height function with the minimal weight, that is the minimal weight of an edge to the power of total number of edges in triangles of the second type (let us denote it by $S_1$).

\bb
Z(\Omega_n,\chi_n)=\prod_{j}{Z(T^{j},\chi_{\rho}^{j})}\geq Z_L:=\prod_{k} {Z(T^{k},\chi_{\rho}^{k})}w_{\min}^{S_1}
\ee

Let us estimate the triangles of the first type using triangle case to count $\delta/3$-close height functions to make sure that we include only height functions $\delta$-close to $g$. For triangles of the first type we have the following,

\bb
\lim_{n \to \infty}n^{-2}\log \prod_j Z(T^{j},\chi_{\rho}^{j}|\delta/3)=\sum_k \sigma(s^k,t^k) \times \cA(T^{k})+o(\delta),
\ee

where $(s^k,t^k)$ is a slope of $g_{\ell}$ on the boundary of the triangle $T^k$ and $\cA(T^{k})$ is the area of the triangle $T^k$. 

Finally, the lower bound after taking the limit as $n\to \infty $ is the following, 

\bb
\sum_{j} \sigma(s^j,t^j) \cA(T^{j})+\log{w_{\min}}\cS+o(\delta),
\ee
where we bound the total number of edges in triangles of the second type by the area of such triangles that we denote by $\cS$\footnote{Note that $\cS\to 0$ as $\delta\to 0$ because the fraction of the triangles of the second type is $o(\delta)$.}.

\subsubsection{The upper bound.}

We can use almost the same strategy to make an upper bound. First, we have to include all height function $\delta$-close to $g$. Let us estimate the triangles of the first type by the same way, but count $3\delta$-close height functions. For the triangles of the second type we make a rough estimate taking the \textit{maximal weight} of an edge to the power of the number of edges in this triangles multiplied by a number of summands in the cutting rule that is $2^{o(n)}$.

\bb
Z(\Omega_n,\chi_n)\leq Z_U:=\prod_{j}{Z(T^{j},\chi_{\rho}^{j})}=\prod_{j}^{\prime}{Z(T^{j},\chi_{\rho}^{j})}w_{\max}^{S(T_j)}2^{o(n)}
\ee
And after taking limit as $n\to \infty$ the normalized upper bound is the following,

\bb
\sum_{j}^{\prime} \sigma(s^j,t^j) \cA(T^{j})+\log{w_{\max}}\cS+o(\delta).
\ee

Both bounds are $\mathcal{F}(g_{\ell})+o(\delta)$\footnote{$\int_{\Omega}{\sigma(\nabla g_{\ell}) dx dy}=\sum_{j}{\sigma(s_j,t_j)}$} that differs from $\mathcal{F}(g)$ by $o(\delta)$ from the lemma 2.2 from \cite{CKP}.
Thus, after taking the limit as $\delta \to 0$ we have the theorem.

\end{proof}

\section*{Appendix: Absolute height functions}

\subsection{Ribbon graph structure}

Every surface graph is also a ribbon graph: a ribbon graph is a graph with an additional structure given by, for each vertex, a cyclic order of the edges at this vertex. For the case of bipartite graphs we can fix an orientation simply orient edges at every black vertex clock-wise and at white vertices counter clock-wise.

An oriented path on a ribbon graph is an oriented path on a dual graph that preserves cyclic order at every vertex i.e. it goes in clock-wise direction at every black vertex and counter-clock-wise at white one. See an example on \href{zigzag}{Figure \ref{zigzag}}.

\begin{figure}[h!]
\centering
\includegraphics[width=0.5\textwidth]{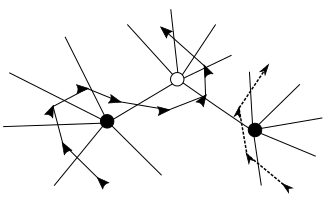}
\caption{Examples of oriented paths.}
\label{zigzag}
\end{figure}

\subsection{Absolute height functions}

Due to the dependence of $h_{D,D^{\prime}}$ and $N_{\GT}$ on $D^{\prime}$ it is more convenient to define an absolute homology class for $D$. One way to do it is to fix a 1-chain $\Phi$ such that $[D] − \Phi$ is a cycle.
\label{abs}

It can be done for every bipartite graph without self-intersections, see theorem 3.3 on page 28 in \cite{GK}.
We will focus on regular graphs, i.e. graphs with the same valence $\cN$ for all vertices.

In this case one can pick $\Phi=\sum_{e}\frac{1}{\cN} [e].$
Obtained height functions are called \textit{absolute height functions} \footnote{However, one needs to change coefficients of homology from $\ZZ$ to $\mathbb{Q}$ or to multiply coefficients by $\cN$ to make them integers.}.

One can notice that absolute height functions obey the local rule. It states that around each vertex they increase with the respect to cyclic order around this vertex: around black vertices they increase in clock-wise direction and around white vertices in counter-clock-wise direction. See an example on  \hyperref[fig:localrule]{Figure \ref{fig:localrule} }.

\begin{figure}[h!]
\centering
\includegraphics[width=0.45\textwidth]{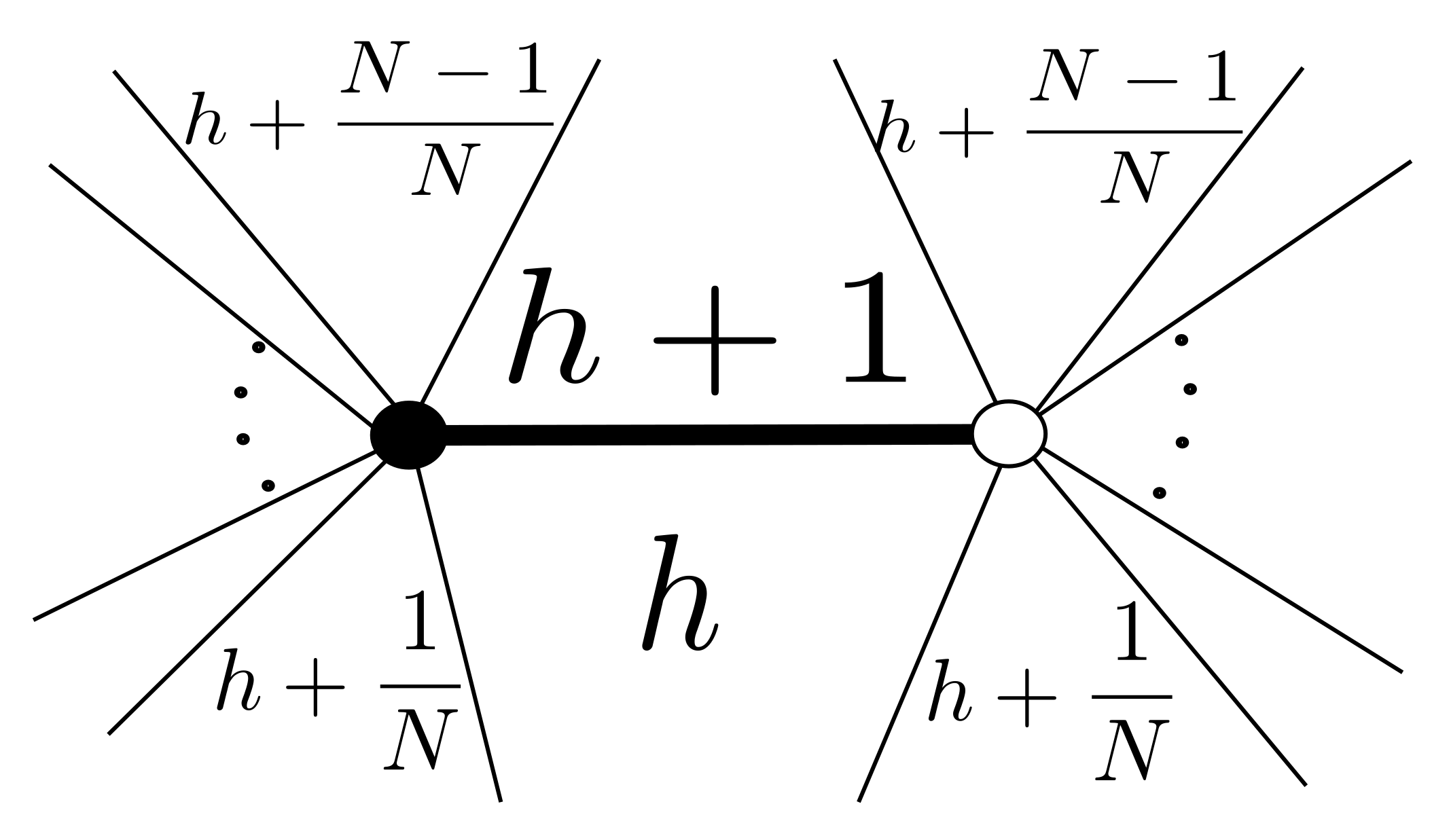}
\caption{The local rule for an absolute height function}
\label{fig:localrule}
\end{figure}

Moreover, it is not hard to check that every function on faces that satisfy the local rule is an absolute height function of some dimer cover.

\begin{proposition}
There is a bijection between absolute height functions and functions on faces
that satisfy the local rule.
\end{proposition}

Also note that absolute height functions are Lipschitz functions in a sense described below, which is a consequence of the local rule.

Let $\pi(f_1,f_2)$  be the length of the shortest oriented path on dual graph of $\Gamma$ that connects faces $f_1$ and $f_2$.

\begin{proposition}
Every absolute height function satisfies modified Lipschitz condition:
\bb
h_D(f_1) - h_D(f_2) \leq \frac{1}{\cN} \pi(f_1,f_2)
\ee
\end{proposition}

\begin{remark}

It is important that there are such functions that satisfy Lipschitz condition, but are not absolute height functions. For example, a constant function on faces satisfies Lipschitz condition, but does not satisfy the local rule.
\end{remark}

\end{document}